\documentclass[10pt, twocolumn]{IEEEtran}

\usepackage{array,xcolor}
\usepackage{booktabs}
\usepackage{multirow}
\usepackage{comment}
\usepackage{setspace}
\renewenvironment{thebibliography}[1]{
  \begin{oldthebibliography}{#1}
    \setlength{\itemsep}{0em}
    \setlength{\parskip}{0em}
}
{
  \end{oldthebibliography}
}
\usepackage{lipsum}
\usepackage{bm}
\usepackage{tikz}
\usetikzlibrary{positioning,calc,arrows.meta,fit,backgrounds}
\usepackage{pgfplots}
\usepackage{pgfplotstable}
\usetikzlibrary{spy}
    
\usepgfplotslibrary{colormaps}
\usepackage{mathrsfs}
\usepackage{mathtools}
\usetikzlibrary{shapes.geometric, arrows}
\usepgfplotslibrary{external}
\usepackage{slashbox}

\usepackage{pifont}
\pgfplotsset{grid style={dashed,gray}}
\pgfplotsset{minor grid style={dotted,gray}}
\pgfplotsset{major grid style={dashed,gray}}

\usepackage{caption}
\usepackage{subcaption}
\usepackage{acronym}
\newacro{5g} [5G] {fifth-generation}
\newacro{6g} [6G] {sixth-generation}
\newacro{isac} [ISAC] {integrated sensing and communication}
\newacro{ul} [UL] {uplink}
\newacro{mu} [MU] {multi-user}
\newacro{csi} [CSI] {channel state information}

\newacro{em} [EM] {electromagnetic}
\newacro{siso} [SISO] {single input single output}
\newacro{miso} [MISO] {multiple input single output}
\newacro{mimo} [MIMO] {multiple input multiple output}
\newacro{xl-mimo} [XL-MISO] {extremely large-scale MIMO}
\newacro{ras} [RAs] {reconfigurable antennas}
\newacro{los} [LoS] {line-of-sight}
\newacro{nlos} [NLoS] {non-line-of-sight}
\newacro{shod} [SHOD] {spherical harmonious orthogonal decomposition}
\newacro{ofdm} [OFDM] {orthogonal frequency division multiplexing}
\newacro{dof} [DoF] {degree of freedom}
\newacro{fim} [FIM] {Fisher information matrix}
\newacro{em} [EM] {electromagnetics}
\newacro{er-fas} [ER-FAS] {electromagnetically reconfigurable FAS}
\newacro{eras} [ERAs] {electromagnetically reconfigurable antennas}
\newacro{aod} [AOD] {angle-of-departure}
\newacro{aoa} [AOA] {angle-of-arrival}
\newacro{aoas} [AOAs] {angles-of-arrival}
\newacro{sp} [SP] {scatter point}
\newacro{ml} [ML] {maximum likelihood}
\newacro{mse} [MSE] {mean square error}
\newacro{snr} [SNR] {signal-to-noise ratio}
\newacro{lmr} [LMR] {line-of-sight to multipath ratio}

\newacro{rmse} [RMSE] {root mean square error}
\newacro{crb} [CRB] {Cram\'er-Rao bound}
\newacro{peb} [PEB] {position error bound}
\newacro{kld} [KLD] {Kullback–Leibler divergence}
\newacro{siso} [SISO] {single-input-single-output}
\newacro{mimo} [MIMO] {multiple-input multiple-output}
\newacro{mcrb} [MCRB] {misspecified Cram\'er-Rao bound}
\newacro{bs} [BS] {base station}
\newacro{ue} [UE] {user equipment}
\newacro{arv} [ARV] {array response vector}
\newacro{upa} [UPA] {uniform planar array}
\newacro{rf} [RF] {radio frequency}
\newacro{bb} [BB] {baseband}
\newacro{ris} [RIS] {reconfigurable intelligent surface}
\newacro{rms} [RMS] {root mean square}
\newacro{psd} [PSD] {positive semidefinite}
\newacro{lmi} [LMI] {linear matrix inequality}
\newacro{bca} [BCA] {block-coordinate ascent}
\newacro{bcd} [BCD] {block-coordinate descent}

\newacro{ma} [MA] {movable antenna}
\newacro{6dma} [6DMA] {six-dimensional movable antenna}
\newacro{pra} [PRA] {pattern reconfigurable antenna}
\newacro{ra} [RA] {reconfigurable antenna}
\newacro{pra_p} [PRAs] {pattern reconfigurable antennas}
\newacro{espar} [ESPAR] {electronically steerable parasitic array radiator}
\newacro{music} [MUSIC] {multiple signal classification}
\newacro{rss} [RSS] {received signal strength}
\newacro{toa} [TOA] {time of arrival}
\newacro{wsn} [WSN] {wireless sensor network}
\newacro{nm} [NM] {Nelder-Mead}
\newacro{ls} [LS] {least-squares}
\newacro{sdp} [SDP] {semidefinite program}
\newacro{iot} [IoT] {internet of things}

\newacro{mpc} [MPC] {multipath component}
\newacro{mpc_p} [MPCs] {multipath components}

\newacro{fas} [FAS] {fluid antenna system}
\newacro{fas_p} [FAS] {fluid antenna systems}
\newacro{sr-fas} [SR-FAS] {spatially reconfigurable FAS}
\newacro{ris} [RIS] {reconfigurable intelligent surface}
\newacro{ris_p} [RISs] {reconfigurable intelligent surfaces}

\usepackage{amsmath, tabularx} % for cases & *
\usepackage{hyperref}
\usepackage{cleveref}
\usepackage{algpseudocode}
\usepackage{ifpdf}
\usepackage{graphicx,amssymb,lineno}
\usepackage{relsize}
\usepackage{filecontents}
\usepackage{lipsum}
\usepackage{color,soul}
\usepackage{algorithm}
\usepackage{color,soul}
\usepackage{cite}

\makeatletter
\def\BState{\State\hskip-\ALG@thistlm}
\makeatother
\makeatletter
\newcommand*{\rom}[1]{\expandafter\@slowromancap\romannumeral #1@}
\makeatother
\allowdisplaybreaks

\graphicspath{ {Images/} }

\usepackage[T1]{fontenc}
\usepackage{float}%avoiding images to be put at the end of document 

\makeatletter
\newcommand{\multiline}[1]{%
  \begin{tabularx}{\dimexpr\linewidth-\ALG@thistlm}[t]{@{}X@{}}
    #1
  \end{tabularx}
}
\makeatother
\usepackage{changepage}
\usepackage{adjustbox}
\makeatother
\makeatother
\title{
Hybrid Codebook Design for Localization Using Electromagnetically Reconfigurable 

Fluid Antenna System
}

\author{
Alireza Fadakar, \emph{Graduate Student Member},
Yuchen Zhang, \emph{Member, IEEE}, 
Hui Chen, \emph{Member, IEEE}, 
Musa Furkan Keskin, \emph{Member, IEEE},
Henk Wymeersch, \emph{Fellow, IEEE},
and
Andreas F. Molisch, \emph{Fellow, IEEE}
\thanks{A. Fadakar and A. F. Molisch are with the Ming Hsieh Department of Electrical and Computer Engineering, University of Southern California, Los Angeles, California, USA; E-mail: \{fadakarg, molisch\}@usc.edu.}
\thanks{Y. Zhang is with the Electrical and Computer Engineering Program, Computer, Electrical and Mathematical Sciences and Engineering (CEMSE), King Abdullah University of Science and Technology (KAUST), Thuwal 23955-6900, Kingdom of Saudi Arabia; E-mail: yuchen.zhang@kaust.edu.sa.}
\thanks{H. Chen, M. F. Keskin, and H. Wymeersch are with the Department of Electrical Engineering, Chalmers University of Technology, Sweden; E-mails: \{hui.chen, furkan, henkw\}@chalmers.se.}

\thanks{
This work is supported, in part, 
by the National Science Foundation (Grants 2229535 and 2106602), 
in part by the KAUST Global Fellowship Program under Award No. RFS-2025-6844, and
in part by the Swedish Research Council (Grants 2022-03007 and 2024-04390).}
}

\usepackage{accents}

\usepackage{amsthm}

\theoremstyle{plain}

\newtheoremstyle{iremark}
  {\topsep}   % ABOVESPACE
  {\topsep}   % BELOWSPACE
  {\upshape}  % BODYFONT
  {0.2in}       % INDENT (empty value is the same as 0pt)
  {\itshape}  % HEADFONT
  {.}         % HEADPUNCT
  {5pt plus 1pt minus 1pt} % HEADSPACE
  {\thmname{#1}\thmnumber{ \itshape#2}\thmnote{ (#3)}} % CUSTOM-HEAD-SPEC

\usepackage{amsthm}

\newtheorem{theorem}{Theorem}
\newtheorem{lemma}[theorem]{Lemma}
\newtheorem{remark}{Remark}

\newtheorem{proposition}{Proposition}

\theoremstyle{definition}
\newtheorem*{proof}{Proof}

\makeatletter
\newcommand*\rel@kern[1]{\kern#1\dimexpr\macc@kerna}
\newcommand*\widebar[1]{%
  \begingroup
  \def\mathaccent##1##2{%
    \rel@kern{0.8}%
    \overline{\rel@kern{-0.8}\macc@nucleus\rel@kern{0.2}}%
    \rel@kern{-0.2}%
  }%
  \macc@depth\@ne
  \let\math@bgroup\@empty \let\math@egroup\macc@set@skewchar
  \mathsurround\z@ \frozen@everymath{\mathgroup\macc@group\relax}%
  \macc@set@skewchar\relax
  \let\mathaccentV\macc@nested@a
  \macc@nested@a\relax111{#1}%
  \endgroup
}
\makeatother

\begin{document}

\maketitle
\begin{abstract}
Electromagnetically reconfigurable fluid antenna systems (ER-FAS) introduce additional degrees of freedom in the electromagnetic (EM) domain by dynamically steering per‑antenna radiation patterns, thereby enhancing power efficiency in wireless links. 
Unlike prior works on spatially reconfigurable FAS, which adjust element positions, ER-FAS provides direct control over each element’s EM characteristics to realize on‑demand beam‑pattern shaping. 
While existing studies have exploited ER-FAS to boost spectral efficiency, this paper explores its application for downlink localization. 
We consider a multiple-input single-output (MISO) system in which a multi‑antenna ER-FAS at the base station serves a single‑antenna user equipment (UE). 
We consider two reconfigurability paradigms: 
(i) a synthesis model where each antenna generates desired beampatterns from a finite set of EM basis functions, and 
(ii) a finite-state selection model in which each antenna selects a pattern from a predefined set of patterns. 
For both paradigms, we formulate the joint baseband (BB) and EM precoder design to minimize the UE position error bound. 
In the synthesis case we derive low-dimensional closed-form expressions for both the BB and EM precoders. 
For the finite-state model we obtain closed-form BB structures and propose a low-complexity block-coordinate-descent algorithm for EM pattern selection. 
Analytical bounds and extensive simulations show that the proposed hybrid designs for ER-FAS substantially improve UE positioning accuracy over traditional non-reconfigurable arrays.
\end{abstract}

\bstctlcite{IEEEexample:BSTcontrol}
\section{Introduction}
\subsection{Background}
The evolution toward \ac{5g} and \ac{6g} networks demands orders‐of‐magnitude increases in data rate and capacity, driving the need for fundamental advances at the physical layer to overcome intrinsic performance limits \cite{wang2025RA,Wong2023fluid}. 
Central to this evolution is the enhancement of \ac{bs} capabilities, since total system throughput mainly depends on the efficient exploitation of limited physical resources \cite{Ying2024Reconfigurable,Chen2023_5G_6G}. 
While time-frequency resource allocation has been extensively optimized, spatial-domain methods, most notably \ac{mimo}, offer additional degrees of freedom. 
However, scaling antenna aperture on mast-mounted installations encounters practical constraints (e.g., structural limits) and becomes impractical in low‐density deployment scenarios. 
\Ac{fas_p} have recently emerged as a pivotal concept to inject new flexibility into the physical layer, positioning them as a promising enabler for \ac{6g} and beyond \cite{Wong2021fas, Wong2023fluid}. 
Unlike conventional arrays with fixed element positions, \ac{fas} can dynamically relocate the antenna element along a fluidic substrate to sample the most favorable spatial channels thereby achieving large spatial-diversity gains. 
Moreover, \ac{fas} can synergize with other frontier technologies such as \ac{ris_p}\cite{fadakar2025mutual}, massive \ac{mimo}\cite{zheng2025trihybrid}, and \ac{isac}\cite{Zhang2025Joint} to further boost network performance. 

%To clearly distinguish spatial and \ac{em} reconfigurability in \ac{fas}, we categorize \ac{fas} architectures into two main categories: 
%spatially reconfigurable, and electromagnetically reconfigurable. 
\ac{fas} can be categorized into two main categories: \ac{sr-fas}, and \ac{er-fas} which are explained next.
\subsubsection{SR-FAS}
\ac{sr-fas} employs port-selection to reposition the radiating element, thereby exploiting spatial diversity while preserving the antenna's intrinsic \ac{em} properties \cite{New2024fluid_Tutorial}. 
\ac{sr-fas} is closely related to \ac{ma} technology \cite{Zhu2025MA_Tutorial,Zack2025tvt,Zack2025TWC}. 
\ac{ma} and \ac{sr-fas} both share the concept of spatial reconfiguration to improve channel conditions and communication performance \cite{New2024fluid_Tutorial}.
Specifically, \ac{ma} can be regarded as a position-only variant of \ac{sr-fas}. 
Hence, in most cases they can be interchangeable due to their similar concepts \cite{zhu2024historical}. 
However, \ac{sr-fas} further supports pixel‐based and shape/size reconfiguration, which is common in \ac{xl-mimo} deployments \cite{New2024fluid_Tutorial}. 
This introduces mutual coupling challenges in which antenna and circuit theory plays an important part to evaluate the performance of \ac{sr-fas}, whereas mechanical or liquid‐based movable designs in \ac{ma} confine coupling to active radiating elements. 
Moreover, \ac{sr-fas} further explore shape or size adaptation for applications such as cognitive radio networks, \ac{iot}, and body area networks \cite{Martinez2022Toward, Borda2019Coupling}.

\subsubsection{ER-FAS}
\ac{er-fas}, also referred to as \ac{ra}, retain a fixed physical location and reshape their metallic patterns or dielectric substrates to provide on-demand control over operating frequency, polarization, and radiation pattern \cite{wang2025RA,Song2019wideband, Rodrigo2014frequency}.
In this paper, we focus on radiation pattern-based \ac{er-fas}.
Compared to \ac{sr-fas}, the radiation properties of the \ac{er-fas} can be dynamically adjusted in real time by altering the fluid state of the \ac{rf} radiator.
There are various fabrication technologies to implement \ac{er-fas}, notably pixel-based parasitic layouts \cite{Zhang2025pixel, liu2025tri}, 
\ac{espar} configurations, consisting of a single active radiator encircled by passive parasitic elements \cite{Han2021Characteristic, liu2025tri}, 
and liquid-metal fluidic implementations \cite{wang2025RA}. 
For instance in \cite{Zhang2025pixel}, a pixel-based \ac{er-fas} with high switching speed is presented, where reconfigurability is realized by changing the interconnection states between pixels \cite{Lotfi2017Printed}. 
The work in \cite{wang2025RA} proposes a novel \ac{er-fas} architecture for wireless communications, in which the radiation pattern of each array element is reconfigurable via software-controlled fluidic actuation  effectively circumventing the mutual coupling effect in pixel-based fabrications.

However, these implementations present practical trade-offs.  
Pixel-based designs benefit from fast, electronically controlled reconfiguration and well-established fabrication processes, but can exhibit unwanted mutual coupling between active and inactive pixels \cite{wang2025RA}.  
By contrast, liquid-metal fluidic implementations can offer continuous geometry control and conformability that help reduce certain coupling effects and enable stretchable or conformal form factors, yet they introduce engineering challenges such as limited reconfiguration speed.  
It is noteworthy that the mathematical models and optimization framework developed in this paper are agnostic to the specific fabrication technology and apply to all existing classes of implementations.  
For consistency with prior work, we use the \ac{er-fas} terminology throughout this paper \cite{wang2025RA}.

\vspace{-0.3cm}
\subsection{Related Works}
\subsubsection{Wireless Localization/Sensing via SR-FAS}
Radio localization has been extensively studied using fixed position antennas. 
More recently, preliminary studies have investigated \ac{aoa} estimation using \ac{sr-fas} systems, which introduce additional spatial degrees of freedom \cite{xu2025future_doa, xu2025fluid_doa}. 
%\cite{FADAKAR2024103382,sparse2025}. 
%In \cite{Zou2017Emitter}, a 2D emitter localization is proposed by collecting \ac{toa} measurements with a moving receiver.  
%In \cite{FADAKAR2024103382, sparse2025}, a passive 3D multi‐source localization scheme is developed via a single mobile receiver with a high‐resolution \ac{aoa} estimation. 
%The work in \cite{sparse2025} extends this framework to sparse‐signal sources, further exploiting the mobility‐induced diversity offered by \ac{ma}.  
%The work \cite{Shao2025exploiting} employs \ac{6dma} at the \ac{bs} to enhance sensing over targeted regions. By dynamically controlling the 3D positions and orientations of distributed antennas, \ac{6dma} fully exploits spatial channel variations, resulting in significant gains in \ac{aoa} estimation accuracy.
%
The advantages of \ac{sr-fas} have been also demonstrated for \ac{isac} systems \cite{Li2025Movable, Lyu2025Movable, Ding2025Movable, Hao2025fluid_ISAC, Wang2024fluid_DRL}. 
The work in \cite{Li2025Movable} demonstrates that dynamically repositioning antenna elements can substantially enhance spectral efficiency, beamforming precision, coverage adaptability, and transmit-power efficiency in \ac{isac} systems. 
%The work in \cite{Lyu2025Movable} studies a bistatic \ac{mu}-\ac{mimo} \ac{isac} system with an \ac{ma}-enabled \ac{bs}, jointly optimizing antenna weights and positions via a fractional-programming framework and a three-stage search-based gradient ascent algorithm, and demonstrates performance gains over fixed-position arrays. 
%In \cite{Ding2025Movable} a near-field \ac{isac} with a full-duplex \ac{bs} equipped with multiple transmit and receive \ac{ma}s is considered, where transmit/receive beamformers, sensing covariance matrices, \ac{ul} power allocation, and \ac{ma} positions are jointly optimized via a two-layer random-position algorithm and a greedy antenna-position matching scheme, yielding significant weighted sum‐rate gains and reduced movement overhead.
The work in \cite{Hao2025fluid_ISAC} jointly optimizes antenna positions in \ac{sr-fas} and dual‐functional beamforming via alternating optimization to maximize sensing \ac{snr} under movement and interference constraints, with robust extensions to imperfect \ac{csi}.
%The authors of \cite{Wang2024fluid_DRL} consider a \ac{mu}–\ac{mimo} downlink \ac{isac} system in which a 2D \ac{sr-fas} at the \ac{bs} jointly selects a subset of antenna ports and designs beamforming vectors via a deep reinforcement learning framework. 
%Their approach significantly increases the user sum‐rate while meeting sensing constraints, even under partial \ac{csi}.

\subsubsection{Wireless Localization/Sensing via ER-FAS}
Existing \ac{er-fas} works for sensing often adopt \ac{espar} architectures.
By dynamically adjusting the reactive loads of the parasitic elements, each element of the \ac{er-fas} can directionally steer its beampatterns, supporting \ac{aoa} estimation \cite{Taillefer2005Direction, Qian2015Direction, Qian2016Direction, Kulas2017Simple, Kulas2018RSS, Tarkowski2019RSS, Bshara2021Noncooperative, Rzymowski2016Single, Groth2023Fast}. 
The work in \cite{Taillefer2005Direction} is one of the earliest applications of \ac{er-fas} for \ac{aoa} estimation using the \ac{music} algorithm. 
In \cite{Qian2015Direction,Qian2016Direction}, the authors develop compressive‐sensing frameworks for \ac{er-fas}-enabled \ac{aoa} estimation via sparse representations, and \cite{Rzymowski2016Single} proposes a novel single‐anchor indoor localization scheme based on \ac{er-fas}. 
In \cite{Kulas2017Simple}, 2D-\ac{aoa} estimation is achieved using only \ac{rss} measurements, while \cite{Kulas2018RSS} demonstrates that interpolating calibrated radiation patterns can enhance \ac{er-fas}-based \ac{aoa} estimation in \ac{wsn} nodes. 
The authors in \cite{Tarkowski2019RSS} further improve performance by applying support‐vector classification to \ac{rss} outputs.  
More recently, \cite{Bshara2021Noncooperative} illustrates rapid optimal beam selection at mmWave frequencies, and \cite{Groth2023Fast} presents a calibration‐free single‐anchor indoor localization algorithm for \ac{bs}s equipped with \ac{er-fas}. 
\vspace{-0.5cm}
\subsection{Motivation}
Although \ac{er-fas} enabled systems have demonstrated promising gains in communication‐centric applications \cite{wang2025RA,liu2025tri,zheng2025trihybrid,Ying2024Reconfigurable}, they still remain in the early stages of development, and their impact on wireless sensing particularly localization has yet to be examined. 
In particular, the existing works on \ac{er-fas} in wireless sensing and localization are mostly limited to single‐antenna \ac{bs}. 
Thus, they cannot readily extend to modern multi-antenna systems in \ac{5g}, \ac{6g}, and beyond, where integrating \ac{fas} requires the joint optimization of \ac{em} and \ac{bb} precoders \cite{Wong2023fluid}.
Moreover, the existing works yield suboptimal performance due to heuristic \ac{em} precoding design without discussion on optimal designs and theoretical bounds. 
To the best of our knowledge, this is the first study to jointly optimize both \ac{bb} and \ac{em} precoders in a \ac{miso} system employing \ac{er-fas} to maximize downlink localization accuracy in a wideband \ac{ofdm} wireless system. 

\subsection{Contributions}
The key contributions are summarized as follows:
\begin{itemize}
\item 
We rigorously derive the signal models and corresponding \ac{fim} for localization with \ac{er-fas} under two reconfigurability paradigms: 
(i) the synthesis model, in which each antenna synthesizes a desired beampattern from a finite set of basis functions; 
and (ii) the finite-state selection model, in which each antenna selects one pattern from a discrete library of physically realizable patterns. 
For both cases, we present theoretical bounds on the achievable localization accuracy. 
\item 
For the synthesis model, assuming perfect knowledge of the \ac{ue} position, we derive the optimal low‐dimensional structures of the \ac{bb} and \ac{em} precoders, which reveal that only three distinct beams are required.  
Leveraging these solutions, we then develop a robust, codebook‐based design that maintains high localization accuracy under \ac{ue}‐position uncertainty. 
\item 
For the finite-state selection model with perfect \ac{ue}-position knowledge, we decouple the \ac{bb} and \ac{em} precoder designs. 
In particular, first, by assuming fixed \ac{em} precoders, we derive the optimal low-dimensional structure of the \ac{bb} precoders and show that only three beams are needed similar to the synthesis case. 
Building on our results from the synthesis case, we then introduce a low-complexity \ac{bcd} algorithm to optimize the \ac{em} precoders for each of these three beams. 
Finally, by integrating these closed‐form and iterative strategies, we develop a robust, codebook‐based scheme that maintains high localization accuracy under \ac{ue}‐position uncertainty.
\item 
We develop a low-complexity \ac{ml} based localization algorithm and perform comprehensive simulations to evaluate the \ac{peb} and localization \ac{rmse}, benchmarking against traditional non‐reconfigurable arrays. 
Our results demonstrate substantial gains compared to existing works utilizing traditional non-reconfigurable arrays under both the synthesis and finite-state selection paradigms. 
For the finite-state selection case, we adopt the physically realizable pattern library obtained in \cite{wang2025RA}.
\end{itemize}

\paragraph*{Notation}
Matrices are denoted by boldface uppercase letters (e.g., $\mathbf{X}$) and vectors by boldface lowercase letters (e.g., $\mathbf{x}$). 
The superscripts $(\cdot)^{\mathsf{T}}$, $(\cdot)^{\mathsf{H}}$, and $(\cdot)^{-1}$ denote transpose, Hermitian (conjugate transpose), and matrix inverse, respectively.  
The expression $[\mathbf{x}_1,\ldots,\mathbf{x}_n]$ denotes horizontal concatenation of vectors $\mathbf{x}_1,\ldots,\mathbf{x}_n$, $\mathrm{diag}(\mathbf{x})$ is the diagonal matrix with the entries of $\mathbf{x}$ on its main diagonal, and $\mathrm{blkdiag}\bigl\{\mathbf{A}_{1},\,\dots,\,\mathbf{A}_{n}\bigr\}$ is a block diagonal matrix with diagonal blocks $\{\mathbf{A}_{i}\}_{i=1}^{n}$. 
The $n\times n$ identity matrix is $\mathbf{I}_n$. 
The notations $\lVert\mathbf{x}\rVert$ and $\lVert\mathbf{x}\rVert_0$ denote the $\ell_2$ (Euclidean) norm and the $\ell_0$ norm (the number of nonzero entries of $\mathbf{x}$), respectively.
We use $\mathbf{A}\otimes\mathbf{B}$, and $\mathbf{A}\odot\mathbf{B}$, for the Kronecker product, and Hadamard product, respectively. 
Finally, $\boldsymbol{\Pi}_{\mathbf{X}}=\mathbf{X}(\mathbf{X}^{\mathsf{H}}\mathbf{X})^{-1}\mathbf{X}^{\mathsf{H}}$ denotes the orthogonal projector onto the column space of $\mathbf{X}$, and $\boldsymbol{\Pi}_{\mathbf{X}}^{\perp}=\mathbf{I}-\boldsymbol{\Pi}_{\mathbf{X}}$.

\section{System and Signal Model}
As shown in Fig.~\ref{fig:system-model}, we consider a \ac{miso} wireless system, comprising an \ac{er-fas} at the \ac{bs} with a \ac{upa} structure consisting of $M=M^{h}M^{v}$ antennas with $M^h$ and $M^v$ being the number of rows and columns, respectively. 
Each element of the \ac{er-fas} can choose one of several predefined states or radiation patterns in the finite-state selection model or synthesize the desired pattern. 
For instance, in Fig.~\ref{fig:heatmap_states}, the heatmaps of $8$ states are shown illustrating the radiation patterns in the finite‐state selection model. 
The details will be provided in Sec.~\ref{sec:reconfigurability_models}. 
Furthermore, we assume a \ac{ue} equipped with a single antenna, and $I$ \ac{mpc_p} to model interference based on a geometric channel model which will be detailed in the next subsection.
It is assumed that the \ac{bs} transmits $N_t$ \ac{ofdm} downlink transmissions to the \ac{ue}, each with $N_s$ subcarriers.  
\begin{figure*}[ht]
\centering
\begin{subfigure}[t]{0.25\textwidth}
\centering
\includegraphics[height=0.19\textheight]{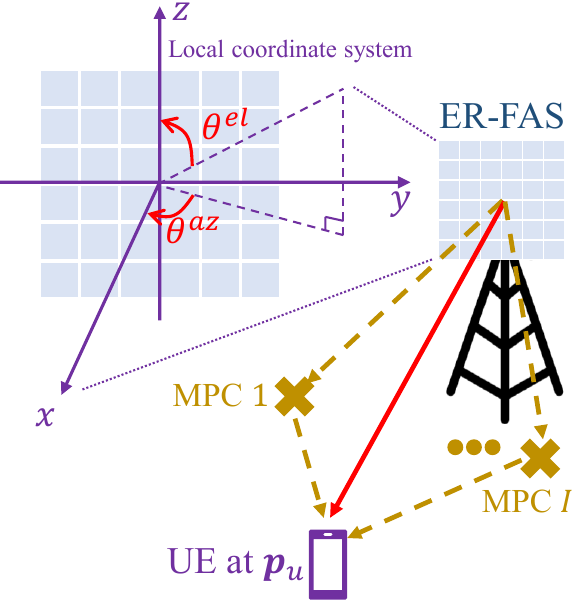}
\caption{}
\label{fig:system-model}
\end{subfigure}%
\hfill
\begin{subfigure}[t]{0.75\textwidth}
\centering
\includegraphics[height=0.21\textheight]{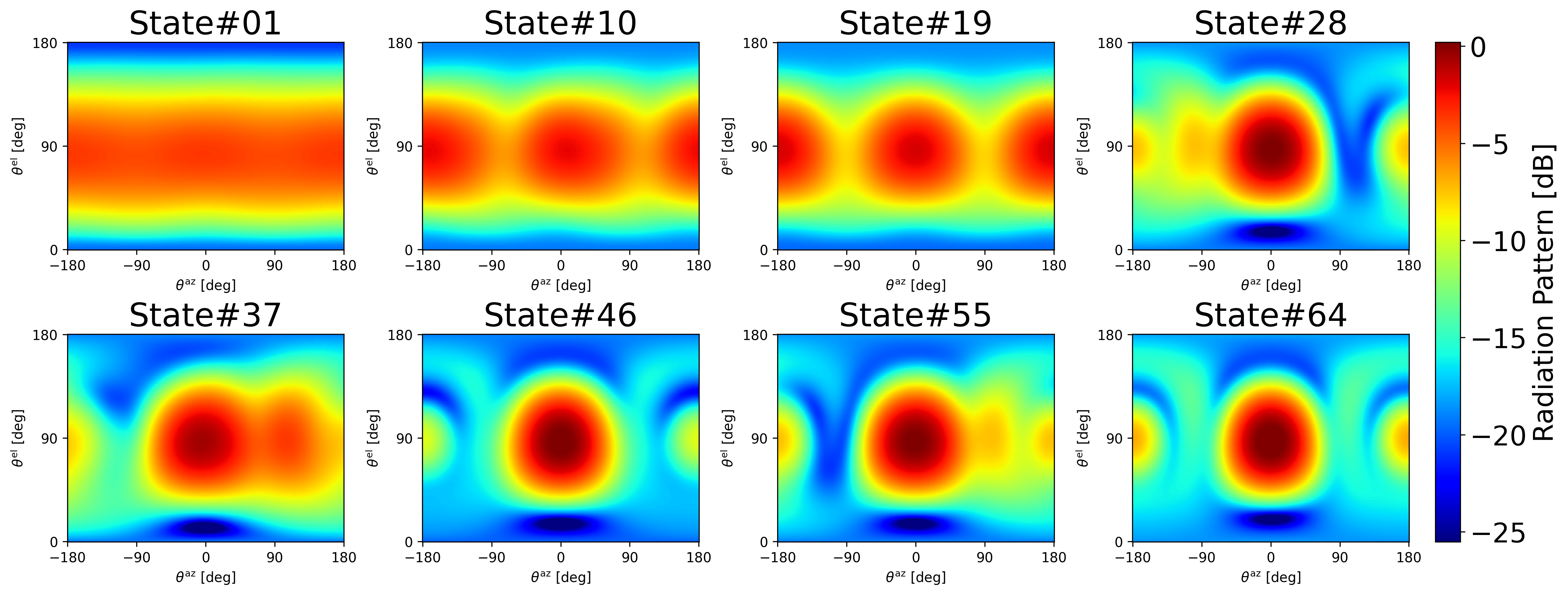}
\caption{}
\label{fig:heatmap_states}
\end{subfigure}
\caption{%
(a) Considered \ac{er-fas} assisted system. This paper aims to jointly optimizing \ac{bb} and \ac{em} precoders to maximize \ac{ue} localization performance.
(b) The 2D heatmap radiation patterns (i.e., $\overline{b}_{s}(\boldsymbol{\theta})$) of $8$ states among the total of $S_{\text{max}}=64$ states. 
In this paper, the physically realizable patterns in \cite{wang2025RA} are used for simulations.
}
\label{fig:system_model_states}
\end{figure*}

The received signal in the $t$-th symbol on the $n$-th subcarrier is given by:  
\begin{equation}\label{eq:y-def}
y_{t,n}
=
\sqrt{P}
\mathbf{h}_{t,n}^{\mathsf{T}}
\mathbf{f}_{t}
+
v_{t,n},
\end{equation}
where $P$ denotes the transmit power, and the vector $\mathbf{f}_{t} \in \mathbb{C}^{M}$ denotes the \ac{bs} \ac{bb} precoder for the $t$-th transmission. 
To keep the transmit power $P$ constant over the entire transmissions the \ac{bb} precoders are assumed to satisfy $\sum_{t=1}^{N_t} \mathbf{f}_t^\mathsf{H}\mathbf{f}_t=1$.
%We set the constraint $\sum \mathbf{f}_t^\mathsf{H}\mathbf{f}_t=1$ so that $P$ become total transmit power. 
%
Moreover, $\mathbf{h}_{t,n} \in \mathbb{C}^{M}$ represents the wireless channel between the \ac{bs} and \ac{ue}, while $v_{t,n} \sim \mathcal{CN}(0,\sigma^2_v)$ denotes the additive white Gaussian noise. 
It is noteworthy that the dependency of the channel on the time index $t$ is due to the time-variant nature of \ac{er-fas} which will be explained later.

Based on the geometric channel model, assuming the presence of a \ac{los} path and $I$ \ac{nlos} paths, the channel vector $\mathbf{h}_{t,n}$ can be expressed as \cite{wang2025RA, Ying2024Reconfigurable, fadakar2024multi}:
\begin{equation}\label{eq:ch-def}
\mathbf{h}_{t,n}
=
\sum_{i=0}^{I}
\alpha_{i}
e^{-j2\pi\tau_in\Delta f}
\mathbf{q}_{t}(\boldsymbol{\theta}_{i}),
\end{equation}
where $i=0$ corresponds to \ac{los} component, 
$\alpha_i = \rho_i e^{j \varphi_i}$ represents the complex gain of the $i$-th path, with $\rho_i$ and $\varphi_i$ denoting its modulus and phase components, respectively. 
Moreover, $\mathbf{q}_{t}(\boldsymbol{\theta}) \in \mathbb{C}^{M}$ represents the joint contribution of the \ac{arv} and the radiation patterns of the antennas at the \ac{bs} in the $t$-th transmission for a given \ac{aod} $\boldsymbol{\theta}$, and is defined as follows:
\begin{equation}\label{eq:q-def}
\mathbf{q}_{t}(\boldsymbol{\theta})
=
\mathbf{g}_{t}(\boldsymbol{\theta})
\odot 
\mathbf{a}(\boldsymbol{\theta}),
\end{equation}
where $\mathbf{g}_{t}(\boldsymbol{\theta}) \in \mathbb{C}^{M}$ represents the complex radiation patterns of the \ac{er-fas} elements, and $\mathbf{a}(\boldsymbol{\theta})$ denotes the \ac{arv} of the corresponding array.

In \eqref{eq:ch-def}, $\tau_i$ denotes the delay of the $i$-th path with $i=0$ accounting for the \ac{los} path:
\begin{equation}
\tau_{0}
=
\frac{\lVert \mathbf{p}_u-\mathbf{p}_b\rVert}{c},\ 
\tau_{i}
=
\frac{\lVert \mathbf{p}_u-\mathbf{p}_s^{(i)}\rVert
+\lVert \mathbf{p}_s^{(i)}-\mathbf{p}_b\rVert}{c},
\end{equation}
where $\mathbf{p}_u\in \mathbb{R}^{3}$ and $\mathbf{p}_b\in \mathbb{R}^{3}$ denote the position of the \ac{ue} and the center of the \ac{bs}, respectively. 
Moreover, $\mathbf{p}_s^{(i)} \in \mathbb{R}^3$ denotes the position of the $i$-th \ac{mpc} characterizing the corresponding \ac{nlos} path.  

The vector $\boldsymbol{\theta}_i = [\theta_i^{\text{el}}, \theta_i^{\text{az}}]^\mathsf{T}$ represents the elevation and azimuth \ac{aod} in the $i$-th path where $i=0$ corresponds to the \ac{los} path. 
These angles can be obtained as follows:
\begin{align}\label{eq:theta_phi_vals}
& \theta_0^{\text{el}}
=
\text{arccos} \left(
\frac{[\mathbf{p}_{u;b}]_{3}}{\lVert\mathbf{p}_{u;b}\rVert}
\right),\ 
\theta_0^{\text{az}}
=
\text{arctan2} (
[\mathbf{p}_{u;b}]_{2},[\mathbf{p}_{u;b}]_{1}
),\\
& 
\theta_i^{\text{el}}
=
\text{arccos} \left(
\frac{[\mathbf{p}_{s;b}^{(i)}]_{3}}{\lVert\mathbf{p}_{s;b}^{(i)}\rVert}
\right),\ 
\theta_i^{\text{az}}
=
\text{arctan2} (
[\mathbf{p}_{s;b}^{(i)}]_{2},[\mathbf{p}_{s;b}^{(i)}]_{1}
),
\end{align}
where 
$\mathbf{p}_{u;b}=\mathbf{R}_b^{\top}(\mathbf{p}_u-\mathbf{p}_b)$, and 
$\mathbf{p}_{s;b}^{(i)}=\mathbf{R}_b^{\top}(\mathbf{p}_s^{(i)}-\mathbf{p}_b)$
with $\mathbf{R}_b\in \mathbb{R}^{3\times 3}$ being the rotation matrix corresponding to the orientation of the \ac{bs}. 
The \ac{arv} $\mathbf{a}(\boldsymbol{\theta})$ is defined as:
\begin{equation}
\mathbf{a}(\boldsymbol{\theta})
=
e^{-j2\pi\omega^{h}\mathbf{k}(M^h)}
\otimes
e^{-j2\pi\omega^{v}\mathbf{k}(M^v)},
\end{equation}
where $\mathbf{k}(M) = [0, \dots, M-1]^\mathsf{T}$, and $\omega^{h}$ and $\omega^{v}$ represent the spatial horizontal and vertical frequencies, respectively.  
As shown in Fig.~\ref{fig:system-model}, if the \ac{upa} is positioned on the YoZ plane of the local Cartesian coordinate system of the \ac{bs}, then  
%(i.e.\ element normals point along $+x$)
\begin{equation}
\omega^{h}
=
d
\sin(\theta^{\text{az}})
\sin(\theta^{\text{el}})/\lambda
,\ 
\omega^{v}
=
d
\cos(\theta^{\text{el}})/\lambda,
\end{equation}
where $d$, denotes the distance between two adjacent elements, and $\lambda$ is the wavelength. 
Moreover, $\theta^{\mathrm{az}},\theta^{\mathrm{el}}$ represent the azimuth and elevation angles in the \ac{bs}’s local Cartesian coordinates. 
Azimuth $\theta^{\mathrm{az}}\in[-\pi,\pi)$ is measured from $+y$ toward $+x$, and elevation $\theta^{\mathrm{el}}\in[0,\pi]$ is the angle with respect to $+z$ direction.
\vspace{-0.2cm}
\section{Modeling Reconfigurable Gains}\label{sec:reconfigurability_models}
The reconfigurability of the complex amplitude of \ac{er-fas} radiation patterns $\mathbf{g}_{t}(\boldsymbol{\theta})$ in \eqref{eq:q-def} can be captured via two complementary modeling paradigms \cite{zheng2025trihybrid}. 
In the first, an idealized \ac{er-fas} is assumed to synthesize arbitrary beampatterns on demand by projecting onto a pre-defined set of orthonormal basis functions \cite{liu2025tri,Ying2024Reconfigurable}. 
In the second model, motivated by practical implementations, each antenna element toggles among a finite collection of discrete states, with each state producing a unique radiation response. For instance, a concrete hardware realization is described in \cite{wang2025RA}. In what follows, we incorporate both the synthesis and the finite-state models into the unified signal framework of \eqref{eq:q-def}.
\vspace{-0.4cm}
\subsection{Synthesis Model}\label{sec:synthesis}
This model assumes that each element of the \ac{er-fas} is capable of synthesizing radiation patterns. 
%an idealization that establishes the theoretical performance upper bound for such antennas and may become increasingly realistic as antenna technologies evolve \cite{zheng2025trihybrid}. 
%It is noteworthy that there are early investigations into orthogonal pattern generation; see, e.g., \cite{Han2021Characteristic}.  
A given beampattern can be described either in the angular domain \cite{liu2025tri} or via a spherical-harmonic expansion \cite{Ying2024Reconfigurable}.  
While the angular-domain formulation often leads to prohibitively large channel matrices, the spherical-domain representation exploits the band-limited nature of practical patterns to achieve a more compact parametrization \cite{zheng2025trihybrid}.  
Consequently, this work also adopts the spherical-domain framework which is detailed next.  

Let $\mathbf{b}(\boldsymbol{\theta}) = [b_{1}(\boldsymbol{\theta}), \dots, b_{Q}(\boldsymbol{\theta})]^{\mathsf{T}}$ denote a vector of $Q$ orthonormal basis functions. 
Hence, the radiation response of the $m$-th antenna element of the \ac{er-fas} at the $t$-th transmission is written as  
\begin{equation}\label{eq:gain-def}
\bigl[\mathbf{g}_{t}(\boldsymbol{\theta})\bigr]_{m}
=
\mathbf{e}_{m,t}^{\mathsf{H}}
\mathbf{b}(\boldsymbol{\theta}),
\end{equation}
where $\mathbf{e}_{m,t}\in\mathbb{C}^{Q}$ is the \ac{em} precoder vector that assigns weights to the basis functions to generate the desired beampattern of the $m$-th element of the \ac{er-fas} during the $t$-th transmission. 
To ensure compliance with the energy conservation law, we assume that $\lVert \mathbf{e}_{m,t}\rVert^2=1$ \cite{Ying2024Reconfigurable} (see Appendix~\ref{app:power_conservation} for proof). 
Let the coefficient matrix $\mathbf{E}_t\!\in\!\mathbb{C}^{M\times MQ}$ be defined as
\begin{equation}\label{def:E_synthesis}
\mathbf{E}_t = \mathrm{blkdiag}\bigl\{\mathbf{e}_{1,t}^{\mathsf{H}},\,\dots,\,\mathbf{e}_{M,t}^{\mathsf{H}}\bigr\}.
\end{equation}
Then, \eqref{eq:q-def} can be rewritten as:
\begin{equation}\label{eq:q_represent}
\mathbf{q}_{t}(\boldsymbol{\theta})
= \mathbf{E}_t\bigl(\mathbf{a}(\boldsymbol{\theta})\otimes\mathbf{b}(\boldsymbol{\theta})\bigr)
= \mathbf{E}_t\,\mathbf{c}(\boldsymbol{\theta}),
\end{equation}
where 
$
\mathbf{c}(\boldsymbol{\theta})
= \mathbf{a}(\boldsymbol{\theta})\otimes\mathbf{b}(\boldsymbol{\theta})
\;\in\;\mathbb{C}^{MQ}.
$
Substituting \eqref{eq:q_represent} into \eqref{eq:ch-def}, then that results into \eqref{eq:y-def}, and finally stacking all $N_s$ symbols within the $t$-th time transmission yields
\begin{equation}\label{eq:y-def-simp}
\mathbf{y}_{t}
=
\sum_{i=0}^{I}
\sqrt{P}\,\alpha_{i}\,
\mathbf{d}(\tau_i)\,
\mathbf{c}(\boldsymbol{\theta}_{i})^{\mathsf{T}}\,
\mathbf{E}_t^{\mathsf{T}}\,
\mathbf{f}_{t}
+
\mathbf{v}_{t}\,,
\end{equation}
where $\mathbf{d}(\tau)\in \mathbb{C}^{N_s}$ denotes the delay steering vector defined as:
\begin{equation}
\mathbf{d}(\tau)=[1,e^{^{-j2\pi \Delta f\tau}},\dots ,e^{-j2\pi (N_s-1)\Delta f\tau}]^{\mathsf{T}}.
\end{equation}
In this paper, we adopt \ac{shod} functions \cite{Costa2010Unified} for basis functions due to their simplicity.
Under the \ac{shod} paradigm, any radiation pattern admits an infinite-series expansion in spherical harmonics \cite{Ying2024Reconfigurable}. 
In this paper, we truncate these bases to the first $Q$ basis functions for simulations.
Further details on the definition of $b_{n}(\boldsymbol{\theta})$ via \ac{shod} are provided in Appendix~\ref{app:shod}.

\begin{remark}
\textbf{Basis‐agnostic methodology:}  
While \ac{shod} bases are used for simulations, our optimization and estimation techniques impose no special structure on these bases. 
Consequently, any complete set of orthonormal functions could be substituted in future work.  Although \ac{shod} may not yet be fully realized in practical hardware, it establishes the ultimate performance bound for \ac{er-fas} implementations \cite{zheng2025trihybrid,Ying2024Reconfigurable}. 
Early efforts, e.g., \cite{Han2021Characteristic, Zhang2025Compact_ESPAR}, have explored the design of physically realizable basis functions (which are also referred to as characteristic modes). 
However, this area is in its early stages and is an important direction for future research.  
\end{remark}
\vspace{-0.5cm}
\subsection{Finite State Selection Model}\label{sec:binary}
The finite‐state selection model provides the most faithful representation of practical hardware implementations \cite{zheng2025trihybrid}. 
We assume a set of $S$ candidate beampatterns, each parameterized by $\boldsymbol{\theta}$ and arranged as  
$
\overline{\mathbf{b}}(\boldsymbol{\theta})
=\bigl[
\overline{b}_{1}(\boldsymbol{\theta}),\dots,\overline{b}_{S}(\boldsymbol{\theta})\bigr]^{\mathsf{T}}
\in\mathbb{R}_{+}^{S}
$.
Accordingly, the radiation response of the $m$-th element of the \ac{er-fas} in the $t$-th transmission is expressed as  
\begin{equation}\label{eq:gain-def-discrete}
\bigl[\mathbf{g}_{t}(\boldsymbol{\theta})\bigr]_{m}
=
\overline{\mathbf{e}}_{m,t}^{\mathsf{T}}
\,
\overline{\mathbf{b}}(\boldsymbol{\theta}),
\end{equation}
where $\overline{\mathbf{e}}_{m,t}\in\{0,1\}^{S}$ is a one‐hot selection vector satisfying $\lVert\overline{\mathbf{e}}_{m,t}\rVert_{0}=1$, i.e., exactly one entry equals $1$. 
Moreover, we assume that the total power of each state equals $1$ i.e., $\int_{0}^{2\pi}\!\int_{0}^{\pi}
\lvert
\overline b_s(\boldsymbol{\theta})
\rvert^2
\;\sin\theta^{\text{el}}\,
d\theta^{\text{el}}\,d\theta^{\text{az}}
=1$.
This constraint ensures the energy conservation law (see Appendix~\ref{app:power_conservation} for proof).

Analogous to the synthesis case, we introduce the discrete‐state coefficient matrix for the $t$-th transmission as:
\begin{equation}
\overline{\mathbf{E}}_{t}
\;=\;
\mathrm{blkdiag}\bigl\{\overline{\mathbf{e}}_{1,t}^{\mathsf{T}},\,\dots,\,\overline{\mathbf{e}}_{M,t}^{\mathsf{T}}\bigr\}
\in\{0,1\}^{M\times MS}.
\end{equation}
Hence, \eqref{eq:q-def} can be represented as
\begin{equation}\label{eq:q_represent_discrete}
\mathbf{q}_{t}(\boldsymbol{\theta})
=
\overline{\mathbf{E}}_{t}\,\bigl(\mathbf{a}(\boldsymbol{\theta})\otimes\overline{\mathbf{b}}(\boldsymbol{\theta})\bigr)
=
\overline{\mathbf{E}}_{t}\,\overline{\mathbf{c}}(\boldsymbol{\theta}),
\end{equation}
where 
$
\overline{\mathbf{c}}(\boldsymbol{\theta})
=
\mathbf{a}(\boldsymbol{\theta})\otimes\overline{\mathbf{b}}(\boldsymbol{\theta})
\;\in\;\mathbb{C}^{MS}
$.
By inserting \eqref{eq:q_represent_discrete} into \eqref{eq:ch-def}, then substituting the result into \eqref{eq:y-def} and collecting the $N_s$ symbols over the $t$-th time frame, we obtain
\begin{equation}\label{eq:y-def-simp-discrete}
\mathbf{y}_{t}
=
\sum_{i=0}^{I}
\sqrt{P}\,\alpha_{i}\,
\mathbf{d}(\tau_{i})\,
\overline{\mathbf{c}}(\boldsymbol{\theta}_{i})^{\mathsf{T}}\,
\overline{\mathbf{E}}_{t}^{\mathsf{T}}\,
\mathbf{f}_{t}
\;+\;
\mathbf{v}_{t}.
\end{equation}

\begin{remark} \textbf{Pattern generation vs.\ state selection:}  
Both the synthesis and finite‐state selection models lead to analogous compact signal representations in \eqref{eq:y-def-simp} and \eqref{eq:y-def-simp-discrete}, respectively, differing only in their physical interpretation and underlying constraints. 
The synthesis model uses continuous coefficients $\{\mathbf{E}_t\}_{t=1}^{N_t}$ to generate beampatterns, whereas the finite-state selection model uses binary selection matrices $\{\overline{\mathbf{E}}_t\}_{t=1}^{N_t}$ to switch among a finite library of predefined patterns.
\end{remark}

\subsection{Fisher Information Analysis}\label{sec:fim}
In this section, we derive the \ac{fim} for the 3D \ac{ue} position and obtain the corresponding \ac{crb}. 
These bounds serve as the objective for the joint design of the \ac{bb} and \ac{em} precoders. 
Owing to the parallel structure of the synthesis and finite‐state selection models introduced in Sec.~\ref{sec:reconfigurability_models}, we present the derivation for the synthesis case only. 
The results for the finite-state selection model can be readily obtained by replacing  
$\mathbf{c}(\boldsymbol{\theta}_{i})$ with $\overline{\mathbf{c}}(\boldsymbol{\theta}_{i})$  
and  
$\mathbf{E}_{t}$ with $\overline{\mathbf{E}}_{t}$.  
\subsubsection{FIM in Channel Domain}\label{sec:fim_ch_domain}
Similar to \cite{fadakar2024multi, fascista2022ris}, and due to the considerable path loss associated with \ac{nlos} components as well as the unknown number of such paths, we treat the \ac{nlos} paths as interference.
Thus, only considering the \ac{los} components with index $0$ in \eqref{eq:y-def-simp}, we define the channel parameter vector
$\boldsymbol{\gamma}\in\mathbb{R}^{5}$ 
as:
\begin{equation}
\boldsymbol{\gamma}
=
[\theta^{\text{el}},\theta^{\text{az}},\tau,\rho,\varphi]^\mathsf{T},
\end{equation}
where the index $0$ are dropped for notational simplicity.
Then, the $(i,j)$-th element of the \ac{fim} $\mathbf{J}_{\boldsymbol{\gamma}}\in\mathbb{C}^{5\times 5}$ can be obtained as \cite{fadakar2025near,fascista2022ris,Zhang2025Joint}:
\begin{equation}\label{eq:fim_def}
[\mathbf{J}_{\boldsymbol{\gamma}}]_{i,j}
=
\frac{2}{\sigma^2_v}
\sum_{t=1}^{N_t}
\Re
\bigg\{
\left(
\frac{\partial \mathbf{x}_t}{\partial \gamma_{i}}
\right)^\mathsf{H}
\left(
\frac{\partial \mathbf{x}_t}{\partial \gamma_{j}}
\right)
\bigg\}
,
\end{equation}
where $\mathbf{x}_t=\sqrt{P}
\alpha\,
\mathbf{d}(\tau)
\mathbf{c}(\boldsymbol{\theta})^\mathsf{T}\mathbf{E}_t^\mathsf{T}
\mathbf{f}_{t}$ denotes the multipath- and noise-free version of \eqref{eq:y-def-simp} at the position $\mathbf{p}_{u}$. 
\subsubsection{FIM in Position Domain}
To compute the \ac{fim} $\mathbf{J}_{\boldsymbol{\eta}} \in \mathbb{R}^{5 \times 5}$ corresponding to the state-space parameter vector 
\begin{equation}\label{eq:gamma}
\boldsymbol{\eta} = \bigl[\mathbf{p}_\text{u}^\mathsf{T},\, \rho,\, \varphi \bigr]^\mathsf{T},   
\end{equation}
we employ a Jacobian-based transformation. 
Specifically, the Jacobian matrix 
$\mathbf{T} = \frac{\partial \boldsymbol{\gamma}^\mathsf{T}}{\partial \boldsymbol{\eta}} \in \mathbb{R}^{5 \times 5}$ 
is used to map the \ac{fim} from the channel parameter domain to the state domain via the relation
$
\mathbf{J}_{\boldsymbol{\eta}} = \mathbf{T} \mathbf{J}_{\boldsymbol{\gamma}} \mathbf{T}^\mathsf{T}.
$
Based on this transformed \ac{fim}, the \ac{peb} at the location $\mathbf{p}_\text{u}$ is given by:
\begin{equation}\label{eq:peb_def}
\text{PEB}
=
\mathrm{tr}([\mathbf{J}_{\boldsymbol{\eta}}^{-1}]_{1:3,1:3}).
\end{equation}

\section{Hybrid Codebook Design: Synthesis Model}\label{sec:synthesis_optimization}
In this section, first we formulate the hybrid beamforming design optimization problem using the synthesis \ac{er-fas} model explained in Sec.~\ref {sec:synthesis}. 
Then, we design a codebook containing $N_t$ beams. 

\subsection{Optimal Structure of Precoders}\label{sec:opt_structure_synthesis}
Assuming a \emph{perfect knowledge} of $\boldsymbol{\eta}$ in \eqref{eq:gamma}, we formulate the optimization problem to jointly design the \ac{em} precoders $\{\mathbf{E}_t\}_{t=1}^{N_t}$ and \ac{bb} precoders $\{\mathbf{f}_t\}_{t=1}^{N_t}$ to minimize the \ac{peb}.
To this end, we formulate the problem as follows:
\begin{subequations} \label{optimization_problem}
\begin{align}
\underset{
\{\mathbf{E}_t,\mathbf{f}_t\}_{t=1}^{N_t}
}
{\textnormal{min}}
\quad  
&
\mathrm{PEB}
\left(
\{\mathbf{E}_t\}_{t=1}^{N_t},
\{\mathbf{f}_t\}_{t=1}^{N_t};\ 
\boldsymbol{\eta}
\right)
\label{def:opt-prob}
\\
\textnormal{s.t.} \quad
&
\lVert [\mathbf{E}_t]_{m,[(m-1)Q+1:mQ]}\rVert^2=1, \label{opt:em_const}
\\
&
\sum_{t=1}^{N_t}\mathbf{f}_t^\mathsf{H}\mathbf{f}_t=1, \label{opt:pow_const}
\\
&
m=1,\dots, M,\ 
t=1,\dots, N_t, \notag
\label{opt-const1}
\end{align}
\end{subequations}
where the constraint \eqref{opt:em_const} is equivalent to the constraint $\lVert \mathbf{e}_{m,t}\rVert^2=1$ discussed in Sec.~\ref{sec:synthesis}. 
The constraint \eqref{opt:pow_const} ensures that $P$ in \eqref{eq:y-def-simp} is the total transmitted power across all transmissions.
The optimization problem \eqref{optimization_problem} is high dimensional due to coupled \ac{em} and \ac{bb} precoders. 
To make \eqref{optimization_problem} more tractable and simpler, we define the vectors $\{\mathbf{w}_t\}_{t=1}^{N_t}$ as follows:
\begin{equation}\label{def:w}
\mathbf{w}_t
=
\mathbf{E}_t^\mathsf{T}
\mathbf{f}_{t}
\in 
\mathbb{C}^{MQ}.
\end{equation}

With this definition, to tackle \eqref{optimization_problem}, we first find the \emph{low-dimensional structure} of the vectors $\{\mathbf{w}_t\}_{t=1}^{N_t}$, and then, we will obtain the optimal structures for the \ac{em} and \ac{bs} precoders. 
First, we state the following lemma:
\begin{lemma}\label{lemma:W_affine}
The \ac{peb} defined in \eqref{eq:peb_def} is a convex function with respect to the matrix $\mathbf{W}=\sum_{t=1}^{N_t}\mathbf{W}_t$ where $\mathbf{W}_t=\mathbf{w}_t\mathbf{w}_t^\mathsf{H}$.
\end{lemma}
\begin{proof}
Assuming 
$\mathbf{W}=\sum_{t=1}^{N_t}\mathbf{W}_t$, the entries of the first row of the \ac{fim} $\mathbf{J}_{\boldsymbol{\gamma}}$ in \eqref{eq:fim_def} can be obtained as:
\begin{align*}
& 
A
\Re
\{
\mathbf{c}^{(2)}(\boldsymbol{\theta})^\mathsf{T}
\mathbf{W}
\mathbf{c}^{(2)}(\boldsymbol{\theta})^{*}
\},\ 
A
\Re
\{
\mathbf{c}^{(3)}(\boldsymbol{\theta})^\mathsf{T}
\mathbf{W}
\mathbf{c}^{(2)}(\boldsymbol{\theta})^{*}
\},\\
& 
\frac{A\dot{N}_s}{N_s}
\Re
\{
\mathbf{c}^{(1)}(\boldsymbol{\theta})^\mathsf{T}
\mathbf{W}
\mathbf{c}^{(2)}(\boldsymbol{\theta})^{*}
\},\ 
\frac{A}{\rho}
\Re
\{
\mathbf{c}^{(1)}(\boldsymbol{\theta})^\mathsf{T}
\mathbf{W}
\mathbf{c}^{(2)}(\boldsymbol{\theta})^{*}
\},\\
&
A
\Re
\{
j
\mathbf{c}^{(1)}(\boldsymbol{\theta})^\mathsf{T}
\mathbf{W}
\mathbf{c}^{(2)}(\boldsymbol{\theta})^{*}
\},
\end{align*}
where $A=\frac{2PN_s\rho^2}{\sigma^2_v}, \dot{N}_s=\mathbf{d}(\tau)^\mathsf{H}\frac{\partial\mathbf{d}(\tau)}{\partial\tau}$, and:
\begin{equation}\label{def:c_partials}
\mathbf{c}^{(1)}(\boldsymbol{\theta})
=
\mathbf{c}(\boldsymbol{\theta})
,\ 
\mathbf{c}^{(2)}(\boldsymbol{\theta})
=
\frac{
\partial\mathbf{c}(\boldsymbol{\theta})
}{\partial\theta^{\text{el}}}
,\ 
\mathbf{c}^{(3)}(\boldsymbol{\theta})
=
\frac{
\partial\mathbf{c}(\boldsymbol{\theta})}{\partial\theta^{\text{az}}}.
\end{equation}

The elements in other rows can be obtained in a similar manner. 
Hence, the elements of the \ac{fim} $\mathbf{J}_{\boldsymbol{\gamma}}$ defined in \eqref{eq:fim_def} linearly depend on $\mathbf{W}$. 
Thus, according to the composition rule \cite{boyd2004convex}, the \ac{peb} in \eqref{eq:peb_def} is a convex function of this matrix.
\end{proof}

Thus, to transform the non-convex optimization problem \eqref{optimization_problem} into a convex one, first we represent it in terms of $\mathbf{W}$, and assume that there exists feasible \ac{bs} and \ac{em} precoders $\{\mathbf{E}_t,\mathbf{f}_t\}_{t=1}^{N_t}$ for the optimal matrix $\mathbf{W}$. 
Hence, with this assumption, we drop the constraints \eqref{opt:em_const}, and \eqref{opt:pow_const}, and reformulate \eqref{optimization_problem} as follows:
\begin{subequations} \label{opt_prob_represent}
\begin{align}
\underset{
\mathbf{W}
}
{\textnormal{min}}
\quad  
&
\mathrm{PEB}
\left(
\mathbf{W};
\boldsymbol{\eta}
\right)
\label{def:opt-prob-represent}
\\
\textnormal{s.t.} \quad
&
\mathrm{tr}\left(
\mathbf{W}
\right)=1
, \label{opt:pow_const_w}
\\
& 
\mathrm{rank}(\mathbf{W})\le N_t \,,
\label{opt:rank_1_constraints}
\end{align}
\end{subequations}
where \eqref{opt:pow_const_w} is the equivalent power constraint to \eqref{opt:pow_const} 
\footnote{
Since $\lVert \mathbf{e}_{m,t}\rVert^2=1$, we have 
$
\mathrm{tr}(\mathbf{W}_t)=\mathbf{w}_t^\mathsf{H}\mathbf{w}_t
=
\mathbf{f}_t^\mathsf{H}\mathbf{f}_t
$. 
Combining this with \eqref{opt:pow_const}, we obtain $\mathrm{tr}(\mathbf{W})=\sum_{t=1}^{N_t}\mathrm{tr}(\mathbf{W}_t)= \sum_{t=1}^{N_t}
\mathbf{f}_t^\mathsf{H}\mathbf{f}_t=1$.
}.
Then, we relax the problem by dropping the rank constraint in \eqref{opt:rank_1_constraints} to obtain a convex optimization problem. 
To achieve low-complexity optimization, we use the low-dimensional structure of the optimal matrix $\mathbf{W}$ in the following proposition.
\begin{proposition}\label{prop:opt_structure}
The optimal matrix $\mathbf{W}$ can be represented as:
\begin{equation}\label{eq:w_low_dim_structure}
\mathbf{W}
=
\mathbf{C}_w
\mathbf{\Xi}
\mathbf{C}_w^\mathsf{H}
\end{equation}
where $\mathbf{\Xi}\in\mathbb{C}^{3\times 3}$ is a \ac{psd} matrix, and
$
\mathbf{C}_w
=
[
\mathbf{c}^{(1)}(\boldsymbol{\theta})^{*},
\mathbf{c}^{(2)}(\boldsymbol{\theta})^{*},
\mathbf{c}^{(3)}(\boldsymbol{\theta})^{*}
]\in \mathbb{C}^{MQ\times 3}
$.
\end{proposition}
\begin{proof}
See Appendix~\ref{app:opt_structure_proof}.
\end{proof}

%In the following proposition, we provide the optimal structure of  \ac{bb} and \ac{em} precoders to achieve optimal value of optimization problem \eqref{opt_prob_represent}.
\begin{proposition}\label{prop:opt_precoders}
Based on the optimal low-dimensional structure of the matrix $\mathbf{W}=\sum_{t=1}^{N_t}\mathbf{W}_t$ in Prop.~\ref{prop:opt_structure}, we propose to use the following $N_t=3$ codewords $\{\mathbf{w}_i\}_{i=1}^{3}$ to achieve the approximate optimal value of \eqref{opt_prob_represent}:
\begin{equation}\label{eq:w_opt_synthesis}
\widehat{\mathbf{w}}_i=
\sqrt{\delta_i}
\frac{\mathbf{c}^{(i)}(\boldsymbol{\theta})^{*}}{\lVert \mathbf{c}^{(i)}(\boldsymbol{\theta})\rVert},
\end{equation}
where $1\ge \delta_i\ge 0$ is the portion of power dedicated to the $i$-th codeword which should satisfy the inequality $\delta_1+\delta_2+\delta_3\le 1$.
For $i\in \{1,2,3\}$, the corresponding \ac{bb} and \ac{em} precoders can be obtained as:
\begin{align}\label{optimal_EM_f_synthesis}
[\hat{\mathbf{f}}_{i}]_m
=
\sqrt{\delta_i}
\frac{\lVert\mathbf{c}_{m}^{(i)}(\boldsymbol{\theta})\rVert}{\lVert \mathbf{c}^{(i)}(\boldsymbol{\theta})\rVert}
e^{j \psi_{m,i}}
,\ 
\hat{\mathbf{e}}_{m,i}
=
\frac{\mathbf{c}_{m}^{(i)}(\boldsymbol{\theta})}{\lVert \mathbf{c}_{m}^{(i)}(\boldsymbol{\theta})\rVert}
e^{-j \psi_{m,i}},
\end{align}
where $\mathbf{c}^{(i)}_m(\boldsymbol{\theta})=[\mathbf{c}^{(i)}(\boldsymbol{\theta})]_{((m-1)Q+1):mQ}$, and $\{\psi_{m,i}\}_{m=1,i=1}^{M,3}$ are arbitrary phases.
\begin{comment}
\begin{align*}
\mathbf{c}_{m}^{(1)}(\boldsymbol{\theta})
& =
[\mathbf{a}(\boldsymbol{\theta})]_{m}
\mathbf{b}(\boldsymbol{\theta}),\\
\mathbf{c}_{m}^{(2)}(\boldsymbol{\theta})
& 
=
[\dot{\mathbf{a}}(\boldsymbol{\theta})]_{m}
\mathbf{b}(\boldsymbol{\theta})
+
[\mathbf{a}(\boldsymbol{\theta})]_{m}
\dot{\mathbf{b}}(\boldsymbol{\theta}),\\
\mathbf{c}_{m}^{(3)}(\boldsymbol{\theta})
& 
=
[\mathbf{a}^{'}(\boldsymbol{\theta})]_{m}
\mathbf{b}(\boldsymbol{\theta})
+
[\mathbf{a}(\boldsymbol{\theta})]_{m}
\mathbf{b}^{'}(\boldsymbol{\theta}),
\end{align*}
for $m=1,\dots ,M$. 
\end{comment}
\end{proposition}
\begin{proof}
See Appendix~\ref{app:opt_precoders}.
\end{proof}

\subsection{Proposed Codebook}\label{sec:codebook_synthesis}
In Sec.~\ref{sec:opt_structure_synthesis}, we obtained a low-complexity joint \ac{bb} and \ac{em} codebook design based on the optimal solutions of \eqref{optimization_problem}, that minimizes the \ac{peb} under perfect knowledge of state parameters including \ac{ue} position. However, in practice the \ac{ue} position may not be available due to measurement errors. 
Let $\bm{\mathcal{U}}$ denote a 3D region accounting for the \ac{ue} position uncertainty region. 
To obtain a robust design under \ac{ue} position uncertainty, assuming $\{\boldsymbol{\theta}_\ell\}_{\ell=1}^{L}$ are uniformly spaced \ac{aod}s covering the uncertainty region, based on Prop.~\ref{prop:opt_precoders}, we propose the following codebook consisting of $N_t=3L$ codewords for the \ac{bb} and \ac{em} precoders:
\begin{align}
\bm{\mathcal{F}}
& =
\bigg\{
\sqrt{\delta_{3\ell -2}} 
\hat{\mathbf{f}}_1(\boldsymbol{\theta}_\ell),
\sqrt{\delta_{3\ell -1}} 
\hat{\mathbf{f}}_2(\boldsymbol{\theta}_\ell),
\sqrt{\delta_{3\ell}} 
\hat{\mathbf{f}}_3(\boldsymbol{\theta}_\ell)
\bigg\}_{\ell=1}^{L},
\notag \\ \label{def:synthesis_codebook}
\bm{\mathcal{E}}
& =
\bigg\{
\widehat{\mathbf{E}}_{1}(\boldsymbol{\theta}_\ell),
\widehat{\mathbf{E}}_{2}(\boldsymbol{\theta}_\ell),
\widehat{\mathbf{E}}_{3}(\boldsymbol{\theta}_\ell)
\bigg\}_{\ell=1}^{L},
\end{align}
where $\widehat{\mathbf{E}}_{i}$ is obtained using \eqref{def:E_synthesis}.
The equivalent representation of this codebook based on vectors $\widehat{\mathbf{w}}_i$ in \eqref{eq:w_opt_synthesis} can be expressed as:
\begin{equation}\label{def:synthesis_codebook_w}
\bm{\mathcal{W}}
=
\bigg\{
\sqrt{\delta_{3\ell -2}} 
\widehat{\mathbf{w}}_1(\boldsymbol{\theta}_\ell),
\sqrt{\delta_{3\ell -1}} 
\widehat{\mathbf{w}}_2(\boldsymbol{\theta}_\ell),
\sqrt{\delta_{3\ell}} 
\widehat{\mathbf{w}}_3(\boldsymbol{\theta}_\ell)
\bigg\}_{\ell=1}^{L}.
\end{equation}

In \eqref{def:synthesis_codebook} and \eqref{def:synthesis_codebook_w}, we separate the power‐allocation coefficients to emphasize that they are the only remaining parameters to be optimized, and to streamline the exposition of their optimization procedure in the next subsection.

\subsection{Power Allocation Optimization}\label{sec:power_alloc_synthesis}
The power allocation optimization problem inside an uncertainty region $\bm{\mathcal{U}}$ can be formulated as:
\begin{subequations} \label{opt:power_alloc}
\begin{align}
\underset{
\{\delta_t\}_{t=1}^{N_t}
}
{\textnormal{min}}
\,
\underset{
\bm{p}_{u}\in\bm{\mathcal{U}}
}
{\textnormal{max}}
\quad  
&
\mathrm{PEB}
\left(
\{
\delta_t\widehat{\mathbf{W}}_t
\}_{t=1}^{N_t}
;\ 
\boldsymbol{\eta}(\mathbf{p}_{u})
\right)
\\
\textnormal{s.t.} \quad
&
\delta_t\ge 0,\ 
\sum_{t=1}^{N_t}\delta_t=1,
\label{opt:delta_sum}
\end{align}
\end{subequations}
where $\delta_t\widehat{\mathbf{W}}_t=\delta_t\widehat{\mathbf{w}}_t\widehat{\mathbf{w}}_t^\mathsf{H}$ is the covariance matrix of the $t$-th codeword in \eqref{def:synthesis_codebook_w}. 
First, we obtain the semi-infinite version of this optimization problem by using its epigraph form \cite{boyd2004convex}:
\begin{subequations} \label{opt:power_alloc_step2}
\begin{align}
\underset{
\varrho,\{\delta_t\}_{t=1}^{N_t}
}
{\textnormal{min}}
\quad  
&
\varrho
\\
\textnormal{s.t.} \quad
&
\mathrm{PEB}
\left(
\{
\delta_t\widehat{\mathbf{W}}_t
\}_{t=1}^{N_t}
;\ 
\boldsymbol{\eta}(\mathbf{p}_{u})
\right)
\le \varrho,\ 
\forall \mathbf{p}_u\in\bm{\mathcal{U}},
\\
& 
\eqref{opt:delta_sum}
.
\end{align}
\end{subequations}
To solve this semi-infinite optimization problem, we discretize $\bm{\mathcal{U}}$ into $N_u$ points $\{\mathbf{p}_{u,i}\}_{i=1}^{N_u}$, and obtain the following approximated problem using Schur complement \cite{boyd2004convex}:
\begin{subequations} \label{opt:power_alloc_step3}
\begin{align}
& 
\underset{
\varrho,
\{\delta_t\}_{t=1}^{N_t}, 
\{u_{i,m}\}_{i=1,m=1}^{N_u,3}
}
{\textnormal{min}}
\varrho
\\
\textnormal{s.t.} \quad
&
\begin{bmatrix}
\mathbf{J}_{\bm{\eta}}\left(
\{
\delta_t\widehat{\mathbf{W}}_t
\}_{t=1}^{N_t}
;\ 
\boldsymbol{\eta}(\mathbf{p}_{u,i})\right) & [\mathbf{I}_5]_{:,m}\\
[\mathbf{I}_5]_{:,m}^{\mathsf{T}} & u_{i,m}
\end{bmatrix}
\succeq 0,
\label{opt:lmi_ineq}
\\
& 
\sum_{m=1}^{3}u_{i,m} \le \varrho,
\\
& 
\eqref{opt:delta_sum}, \notag
\\
&
m=1,2,3,\ i=1,\dots, N_u.
\notag
\end{align}
\end{subequations}
From the proof of Lemma~\ref{lemma:W_affine}, each element of the \ac{fim} is an affine function of the codeword covariance matrices $\{\delta_{t}\widehat{\mathbf{W}}_{t}\}_{t=1}^{N_{t}}$. 
Hence, the constraint in \eqref{opt:lmi_ineq} is a \ac{lmi} in the variables $\{\delta_{t}\}_{t=1}^{N_{t}}$ and $\{u_{i,m}\}$, and the power-allocation problem in \eqref{opt:power_alloc_step3} is a convex \ac{sdp} \cite{boyd2004convex} in the coefficients $\{\delta_{t}\}_{t=1}^{N_t}$. 
Consequently, it can be efficiently solved using convex optimization solvers.  
\vspace{-0.7cm}
\subsection{Complexity of Power Optimization}
According to \cite[Ch.~11]{nemirovski2004interior}, \cite{Keskin2022Optimal}, the \ac{sdp} has the following complexity:
$
O\left(U^2\sum_{i=1}^{K}n_i^2
+
U\sum_{i=1}^{K}n_i^3
\right),
$
where $U$, $K$ denote the number of optimization variables and the number of \ac{lmi} constraints, respectively, and $n_i$ is the row/column width of the $i$-th \ac{lmi} matrix. 
In the \ac{sdp} \eqref{opt:power_alloc}, we have $U=3N_u+N_t+1$ and $K=3N_u$. 
Thus, assuming $N_u\gg N_t$, the complexity of \eqref{opt:power_alloc} is approximately $O(N_u^3)$. 

\vspace{-0.2cm}
\section{Hybrid Codebook Design: Finite State Selection Model}
In this section, we investigate the same problem discussed in Sec.~\ref{sec:synthesis_optimization} using the finite-state selection \ac{er-fas} model explained in Sec.~\ref{sec:binary}. Under the perfect knowledge of the state parameters $\boldsymbol{\eta}$, the problem can be formulated as follows:
\begin{subequations} \label{optimization_problem_d}
\begin{align}
\underset{
\{\overline{\mathbf{E}}_{t}\}_{t=1}^{N_t},\{\mathbf{f}_t\}_{t=1}^{N_t}
}
{\textnormal{min}}
\quad  
&
\mathrm{PEB}
\left(
\{\overline{\mathbf{E}}_{t}\}_{t=1}^{N_t},
\{\mathbf{f}_t\}_{t=1}^{N_t};\ 
\boldsymbol{\eta}
\right)
\label{def:opt-prob_d}
\\
\textnormal{s.t.} \quad
&
\lVert [\overline{\mathbf{E}}_t]_{m,[(m-1)Q+1:mQ]}\rVert_0=1,
\label{opt:em_const_d}
\\
&
\overline{\mathbf{E}}_{t}\in\{0,1\}^{M\times MS}
\notag
\\
&
\sum_{t=1}^{N_t}\mathbf{f}_t^\mathsf{H}\mathbf{f}_t=1, \label{opt:pow_const_d}
\\
&
m=1,\dots, M,\ 
t=1,\dots, N_t, \notag
\end{align}
\end{subequations}
where the constraint \eqref{opt:em_const_d} is equivalent to unit norm constraint $\lVert \overline{\mathbf{e}}_{m,t}\rVert_0=1$ discussed in Sec.~\ref{sec:binary}, and \eqref{opt:pow_const_d} is the power constraint.
Problem \eqref{optimization_problem_d} is non-convex with a large number of coupled discrete and continuous variables. 
To render \eqref{optimization_problem_d} more tractable and simpler, similar to Sec.~\ref{sec:opt_structure_synthesis}, we define the vectors $\{\overline{\mathbf{w}}_t\}_{t=1}^{N_t}$ as follows:
\begin{equation}\label{def:w_d}
\overline{\mathbf{w}}_t
=
\overline{\mathbf{E}}_t^\mathsf{T}
\mathbf{f}_{t}
\in 
\mathbb{C}^{MS}.
\end{equation}
In a similar manner to Sec.~\ref{sec:opt_structure_synthesis}, the following equivalent problem can be obtained:
\begin{subequations} \label{opt_prob_represent_d}
\begin{align}
\underset{
\overline{\mathbf{W}}
}
{\textnormal{min}}
\quad  
&
\mathrm{PEB}
\left(
\overline{\mathbf{W}};
\boldsymbol{\eta}
\right)
\label{def:opt-prob-represent_d}
\\
\textnormal{s.t.} \quad
&
\mathrm{tr}\left(
\overline{\mathbf{W}}
\right)=1
, \label{opt:pow_const_w_d}
\\
& \mathrm{rank}(\overline{\mathbf{W}})\le N_t,
\label{opt:rank_1_constraints_d}
\end{align}
\end{subequations}
where the power constraint is equivalent to \eqref{opt:pow_const_d}
\footnote{From \eqref{def:w_d}, we have $\lVert\overline{\mathbf{w}}_t\rVert
=
\lVert\overline{\mathbf{E}}_t^\mathsf{T}
\mathbf{f}_{t}\rVert=\lVert\mathbf{f}_t\rVert
$. Hence, 
$\mathrm{tr}\left(
\overline{\mathbf{W}}
\right)=\sum_{t=1}^{N_t}\mathrm{tr}\left(
\overline{\mathbf{W}}_t
\right)=
\sum_{t=1}^{N_t}
\lVert
\overline{\mathbf{w}}_t
\rVert^2
=
\sum_{t=1}^{N_t}
\lVert
\mathbf{f}_t
\rVert^2=1
$.
},
which can be transformed to a convex problem by dropping the rank constraint in \eqref{opt:rank_1_constraints_d}. 
Hence, due to an analogy between problems \eqref{opt_prob_represent} and \eqref{opt_prob_represent_d}, for a given 2D-\ac{aod} $\boldsymbol{\theta}$, the following three codewords are sufficient to achieve the optimal value:
\begin{equation}\label{eq:optimal_w_d}
\widehat{\overline{{\mathbf{w}}}}_i=
\sqrt{\delta_i}
\frac{\overline{\mathbf{c}}^{(i)}(\boldsymbol{\theta})^{*}}{\lVert \overline{\mathbf{c}}^{(i)}(\boldsymbol{\theta})\rVert},
\end{equation}
where
\begin{equation}\label{def:c_partials_d}
\overline{\mathbf{c}}^{(1)}(\boldsymbol{\theta})
=
\overline{\mathbf{c}}(\boldsymbol{\theta})
,\ 
\overline{\mathbf{c}}^{(2)}(\boldsymbol{\theta})
=
\frac{
\partial\overline{\mathbf{c}}(\boldsymbol{\theta})
}{\partial\theta^{\text{el}}}
,\ 
\overline{\mathbf{c}}^{(3)}(\boldsymbol{\theta})
=
\frac{
\partial\overline{\mathbf{c}}(\boldsymbol{\theta})}{\partial\theta^{\text{az}}}.
\end{equation}
In \eqref{eq:optimal_w_d}, $\delta_1$, $\delta_2$, and $\delta_3$ are power allocation coefficients.
In the rest of the paper, we will refer to these three codewords by Type-1, Type-2, and Type-3 codewords, respectively, which depend on the 2D-\ac{aod} $\boldsymbol{\theta}$.

%In other words, \eqref{def:w_d} does not admit a closed-form solution for the non-admissible codewords of \eqref{eq:optimal_w_d}. 

Although the three codewords in \eqref{eq:optimal_w_d} achieve the optimum of \eqref{opt_prob_represent_d}, for a given state vector $\boldsymbol{\eta}$, 
the closed-from \ac{bb} and \ac{em} precoders in \eqref{def:w_d} that satisfy these codewords cannot be found because of the $\ell_0$-norm constraint in \eqref{opt:em_const_d} unique to the finite-state selection model.
To address this, we employ an alternating‐optimization procedure to construct admissible codewords $\widetilde{\overline{\mathbf{w}}}_{i}$ whose beampatterns closely approximate those of the optimal vectors $\widehat{\overline{\mathbf{w}}}_{i}$. 
Specifically, we first derive a closed‐form solution for the \ac{bb} precoders with fixed \ac{em} precoders $\{\overline{\mathbf{E}}_{i}\}_{i=1}^{3}$, then optimize the \ac{em} precoders, and finally present a robust codebook design along with the optimization of power-allocation coefficients.
\vspace{-0.2cm}
\subsection{Baseband Precoders Design Under Fixed EM Precoders}
In this subsection, we obtain a closed-form structure of the optimal \ac{bb} precoders for fixed \ac{em} precoders $\{\overline{\mathbf{E}}_{i}\}_{i=1}^{3}$. 
With this assumption, the received interference-free signal at the $t$-th transmission, can be represented as follows:
\begin{equation}\label{eq:y-def-simp-d-v1}
\mathbf{y}_{t}
=
\sqrt{P}\,\alpha\,
\mathbf{d}(\tau)\,
\mathbf{q}_t(\boldsymbol{\theta})^{\mathsf{T}}\,
\mathbf{f}_{t}
\;+\;
\mathbf{v}_{t},
\end{equation}
where $\mathbf{q}_t(\boldsymbol{\theta})$ only depends on $\boldsymbol{\theta}$ due to fixed \ac{em} precoders.
Thus, the optimization problem \textit{under fixed \ac{em} precoders} can be formulated as follows:
\begin{subequations} \label{optimization_problem_d_v1}
\begin{align}
\underset{
\{\mathbf{f}_t\}_{t=1}^{N_t}
}
{\textnormal{min}}
\quad  
&
\mathrm{PEB}
\left(
\{\mathbf{f}_t\}_{t=1}^{N_t};\ 
\{\overline{\mathbf{E}}_{t}\}_{t=1}^{N_t}
,
\boldsymbol{\eta}
\right)
\label{def:opt-prob_fixed_em}
\\
\textnormal{s.t.} \quad
&
\eqref{opt:pow_const_d}.
\notag
\end{align}
\end{subequations}
%In the following proposition, based on the optimal low-dimensional structure of the \ac{bb} precoders under fixed \ac{em} precoders, we propose three \ac{bb} precoders.
\begin{proposition}\label{prop:optimal_precoder_d}
Assuming a perfect knowledge of \ac{ue} position, and based on the optimal low-dimensional structure of the \ac{bb} precoders under fixed \ac{em} precoders, we propose the following three \ac{bb} precoders, which approximately achieve the optimum of \eqref{optimization_problem_d_v1}:
\begin{equation}\label{eq:digital_opt_structure_d}
\hat{\mathbf{f}}_i
=
\sqrt{\delta_i}
\frac{\overline{\mathbf{E}}_i
\overline{\mathbf{c}}^{(i)}(\boldsymbol{\theta})^{*}}{\lVert \overline{\mathbf{E}}_i
\overline{\mathbf{c}}^{(i)}(\boldsymbol{\theta})\rVert}
,
\end{equation}
for $i=1,2,3$, 
where $\delta_i$ is the proportion of the total power dedicated to the $i$-th codeword.
\end{proposition}
\begin{proof}
See Appendix~\ref{app:prop3_proof}.
\end{proof}

\begin{comment}
\subsection{Proposed Codebook Under Fixed EM Precoders}
Similar to Sec.~\ref{sec:codebook_synthesis}, assuming $\{\boldsymbol{\theta}_\ell\}_{\ell=1}^{L}$ are uniformly spaced \ac{aod}s covering the uncertainty region $\mathbf{\mathcal{U}}$, based on Prop.~\ref{prop:optimal_precoder_d}, \textit{under the fixed \ac{em} precoders}, we propose the following optimal codebook consisting of $N_t=3L$ codewords for the \ac{bb} precoders:
\begin{align}\label{def:discrete_codebook_bb}
\bm{\mathcal{F}}
=
\bigg\{
\hat{\mathbf{f}}_1(\boldsymbol{\theta}_\ell),\ 
\hat{\mathbf{f}}_2(\boldsymbol{\theta}_\ell),\ 
\hat{\mathbf{f}}_3(\boldsymbol{\theta}_\ell)
\bigg\}_{\ell=1}^{L}.
\end{align}
\end{comment}

\vspace{-0.5cm}
\subsection{Proposed Heuristic Optimization for EM Precoders}
In Prop.~\ref{prop:optimal_precoder_d}, we proposed three codewords for \ac{bb} precoders under fixed \ac{em} precoders based on optimal low-dimensional solutions of \eqref{optimization_problem_d}. 
In this subsection, we investigate the optimization of the corresponding \ac{em} precoders $\{\overline{\mathbf{E}}_{i}\}_{i=1}^{3}$.
By substituting the optimal structure of the \ac{bb} precoders \eqref{eq:digital_opt_structure_d} into \eqref{def:w_d}, we obtain the following simplified equation:
\begin{align}\label{def:w_d_simp}
\overline{\mathbf{w}}_i(\overline{\mathbf{E}}_i,\boldsymbol{\theta})
=
\sqrt{\delta_i}
\frac{
\overline{\mathbf{E}}_i^\mathsf{T}
\overline{\mathbf{E}}_i
\overline{\mathbf{c}}^{(i)}(\boldsymbol{\theta})^{*}}{\lVert \overline{\mathbf{E}}_i
\overline{\mathbf{c}}^{(i)}(\boldsymbol{\theta})\rVert}
,
\end{align}
for $i\in\{1,2,3\}$, where we have highlighted that the candidate vectors $\overline{\mathbf{w}}_i$ can be viewed as a function of $\overline{\mathbf{E}}_i,\boldsymbol{\theta}$.

For each $i\in\{1,2,3\}$, we optimize the \ac{em} precoding matrix $\overline{\mathbf{E}}_{i}$ so that the resulting beampattern 
of $\overline{\mathbf{w}}_{i}(\overline{\mathbf{E}}_{i},\boldsymbol{\theta})$ best approximates the non‐admissible optimal vector 
$\widehat{\overline{\mathbf{w}}}_{i}$ in \eqref{eq:optimal_w_d}. 
To this end, let 
$
\mathbf{\Theta}\in\mathbb{R}^{N_{g}\times 2}
$
be a discretized 2D‐\ac{aod} grids of $N_{g}$ angle pairs, and let 
$
\overline{\mathbf{C}}\in\mathbb{C}^{MS\times N_{g}}
$
collect the corresponding vectors $\overline{\mathbf{c}}(\boldsymbol{\theta}_n)$ at each grid point i.e., $\overline{\mathbf{C}}=[\overline{\mathbf{c}}(\boldsymbol{\theta}_1),\dots ,\overline{\mathbf{c}}(\boldsymbol{\theta}_{N_g})]$. 
The design of $\overline{\mathbf{E}}_{i}$, which depends on the true 2D-\ac{aod} $\boldsymbol{\theta}$, can then be cast as the following optimization problem:
\begin{equation}\label{optimization_EM_LS}
\widehat{\overline{\mathbf{E}}}_i(\boldsymbol{\theta})
=
\underset{
\overline{\mathbf{E}}_{i}
}
{\textnormal{argmin}}
\quad  
\lVert 
\overline{\mathbf{C}}^\mathsf{T}
\overline{\mathbf{w}}_{i}(\overline{\mathbf{E}}_{i},\boldsymbol{\theta})
-
\overline{\mathbf{C}}^\mathsf{T}
\widehat{\overline{{\mathbf{w}}}}_i
\rVert^2,
\end{equation}
where $\overline{\mathbf{C}}^\mathsf{T}
\overline{\mathbf{w}}_{i}$ 
and 
$\overline{\mathbf{C}}^\mathsf{T}
\widehat{\overline{{\mathbf{w}}}}_i$ 
are the beampatterns of $\overline{\mathbf{w}}_{i}$ 
and 
$\widehat{\overline{{\mathbf{w}}}}_i$ 
evaluated at the grids $\boldsymbol{\theta}_1,\dots ,\boldsymbol{\theta}_{N_g}$, respectively (see \eqref{eq:y-def-simp-discrete} and \eqref{def:w_d}).
After substituting \eqref{eq:optimal_w_d} and \eqref{def:w_d_simp} into \eqref{optimization_EM_LS}, the problem can be represented as:

\begin{equation}\label{optimization_EM_LS_represent}
\widehat{\overline{\mathbf{E}}}_i
(\boldsymbol{\theta})
=
\underset{
\overline{\mathbf{E}}_{i}
}
{\textnormal{argmin}}
\quad  
\underbrace{
\Bigg\lVert 
\frac{
\overline{\mathbf{C}}^\mathsf{T}
\overline{\mathbf{E}}_i^\mathsf{T}
\overline{\mathbf{E}}_i
\overline{\mathbf{c}}^{(i)}(\boldsymbol{\theta})^{*}}{\lVert \overline{\mathbf{E}}_i
\overline{\mathbf{c}}^{(i)}(\boldsymbol{\theta})\rVert}
-
\frac{
\overline{\mathbf{C}}^\mathsf{T}
\overline{\mathbf{c}}^{(i)}(\boldsymbol{\theta})^{*}}{\lVert \overline{\mathbf{c}}^{(i)}(\boldsymbol{\theta})\rVert}
\Bigg\rVert^2
}_{\mathcal{G}_i(\overline{\mathbf{E}}_{i};\ \boldsymbol{\theta})},
\end{equation}
where we have dropped the power allocation coefficient $\delta_i$ as it does not affect the optimization.

Solving \eqref{optimization_EM_LS_represent} by exhaustive search over all $S^{M}$ combinations is computationally expensive when $S$ or $M$ is large. 
Thus, we adopt the \ac{bcd} algorithm detailed in Algorithm~\ref{alg:EM_precoder_opt}. 
Specifically, we initialize the \ac{em} precoder $\overline{\mathbf{E}}_{i}$ and then, at each iteration, sequentially sweep through all $M$ antennas. 
For each antenna, we select the pattern that minimizes the objective function while holding the others fixed. 
The algorithm terminates when either the absolute change in the objective between consecutive iterations falls below a threshold $\epsilon$, or the iteration count exceeds $N_{\max}$.  

\begin{algorithm}
\caption{Proposed BCD algorithm for \ac{em} precoder optimization}\label{alg:EM_precoder_opt}
\begin{algorithmic}[1]
\State 
\multiline{%
\textbf{Inputs}: 
Non-admissible optimal codeword $\widehat{\overline{\mathbf{w}}}_i$ defined in \eqref{eq:optimal_w_d}, 2D‐\ac{aod} $\boldsymbol{\theta}$, matrix $\overline{\mathbf{C}}$, convergence threshold $\epsilon$, max iterations $N_{\max}$.
}
\State 
\multiline{%
\textbf{Output}:
Optimal \ac{em} precoder matrix $\widehat{\overline{\mathbf{E}}}_i$
}
\State 
\multiline{%
\textbf{Initialization}: 
Randomly initialize $\overline{\mathbf{E}}_i$, and set iteration number $r\gets 0$. 
Moreover, set the initial objective function $\mathcal{G}_i^{0}$ using \eqref{optimization_EM_LS_represent}.
}
\While{$\lvert \mathcal{G}^{r}_i-\mathcal{G}^{r-1}_i\rvert>\epsilon$ or $r>N_{\max}$}
\For{Antenna index $m=1:M$}
\State 
\multiline{%
Update the state of the $m$-th antenna to $s$, 

i.e., $[\overline{\mathbf{E}}_i]_{m,((m-1)S+s)}=1$, which would achieve the lowest objective function $\mathcal{G}_i$ in \eqref{optimization_EM_LS_represent}. 
}
\EndFor
\State
\multiline{%
Increment iteration number: $r\gets r+1$.
}
\State
\multiline{%
Calculate the objective function $\mathcal{G}^{r}_i$ using \eqref{optimization_EM_LS_represent}.
}
\EndWhile
\end{algorithmic}
\end{algorithm}
\begin{comment}
\State
\multiline{%
Return the optimal precoder matrix $\widehat{\overline{\mathbf{E}}}_i$.
}
\end{comment}
\vspace{-0.5cm}
\subsection{Proposed Codebook}\label{sec:codebook_finite_state_model}
Similar to Sec.~\ref{sec:codebook_synthesis}, let $\{\boldsymbol{\theta}_\ell\}_{\ell=1}^{L}$ be $L$ uniformly spaced \acp{aod} spanning the uncertainty region $\mathcal{U}$. 
Based on the optimal \ac{bb} precoders in \eqref{eq:digital_opt_structure_d} and the optimal \ac{em} precoders in \eqref{optimization_EM_LS_represent}, we construct a codebook of $N_{t}=3L$ codewords, comprising both \ac{bb} and \ac{em} precoders, as follows:
\begin{align}
\label{def:finite_state_codebook}
\overline{\bm{\mathcal{E}}}
& =
\bigg\{
\widehat{\overline{\mathbf{E}}}_{1}(\boldsymbol{\theta}_\ell),
\widehat{\overline{\mathbf{E}}}_{2}(\boldsymbol{\theta}_\ell),
\widehat{\overline{\mathbf{E}}}_{3}(\boldsymbol{\theta}_\ell)
\bigg\}_{\ell=1}^{L} \,.
\\  \notag 
\overline{\bm{\mathcal{F}}}
& =
\bigg\{
\sqrt{\delta_{3\ell-2}}
\hat{\mathbf{f}}_1(\boldsymbol{\theta}_\ell),
\sqrt{\delta_{3\ell-1}}
\hat{\mathbf{f}}_2(\boldsymbol{\theta}_\ell),
\sqrt{\delta_{3\ell}}
\hat{\mathbf{f}}_3(\boldsymbol{\theta}_\ell)
\bigg\}_{\ell=1}^{L},
\end{align} 

The equivalent representation of this codebook based on optimized feasible vectors $\widetilde{\overline{\mathbf{w}}}_i(\boldsymbol{\theta})=\overline{\mathbf{w}}_i(\widehat{\overline{\mathbf{E}}}_i,\boldsymbol{\theta})$ (obtained by substituting \eqref{optimization_EM_LS_represent} into \eqref{def:w_d_simp}) can be written as:
\begin{equation}\label{def:finite_state_codebook_w}
\overline{\bm{\mathcal{W}}}
=
\bigg\{
\sqrt{\delta_{3\ell -2}} 
\widetilde{\overline{\mathbf{w}}}_1(\boldsymbol{\theta}_\ell),
\sqrt{\delta_{3\ell -1}} 
\widetilde{\overline{\mathbf{w}}}_2(\boldsymbol{\theta}_\ell),
\sqrt{\delta_{3\ell}} 
\widetilde{\overline{\mathbf{w}}}_3(\boldsymbol{\theta}_\ell)
\bigg\}_{\ell=1}^{L} \,.
\end{equation}

In \eqref{def:finite_state_codebook} and \eqref{def:finite_state_codebook_w}, we explicitly factor out the power‐allocation coefficients to underscore that they are the only remaining optimization variables and to simplify the presentation of their optimization in the next subsection.
\vspace{-0.2cm}
\subsection{Power Allocation Optimization}\label{sec:power_alloc_d}
The power allocation optimization problem given an uncertainty \ac{ue} position region $\bm{\mathcal{U}}$ can be formulated as:

\begin{subequations} \label{opt:power_alloc_d}
\begin{align}
\underset{
\{\delta_t\}_{t=1}^{N_t}
}
{\textnormal{min}}
\,
\underset{
\bm{p}_{u}\in\bm{\mathcal{U}}
}
{\textnormal{max}}
\quad  
&
\mathrm{PEB}
\left(
\{
\delta_t\widetilde{\overline{\mathbf{W}}}_t
\}_{t=1}^{N_t}
;\ 
\boldsymbol{\eta}(\mathbf{p}_{u})
\right)
\\
\textnormal{s.t.} \quad
&
\delta_t\ge 0,\ 
\sum_{t=1}^{N_t}\delta_t=1,
\end{align}
\end{subequations}
where $\delta_t\widetilde{\overline{\mathbf{W}}}_t=\delta_t\widetilde{\overline{\mathbf{w}}}_t\widetilde{\overline{\mathbf{w}}}_t^\mathsf{H}$ is the covariance matrix of the $t$-th codeword in \eqref{def:finite_state_codebook_w}.
Problem \eqref{opt:power_alloc_d} is similar to its counterpart for the synthesis scenario \eqref{opt:power_alloc}. 
Thus, this problem is also convex and the same methodologies explained in Sec.~\ref{sec:power_alloc_synthesis} can be applied here.
\vspace{-0.3cm}
\subsection{Complexity of Algorithm~\ref{alg:EM_precoder_opt}}
First, we obtain the complexity of evaluating $\mathcal{G}_i(\overline{\mathbf{E}}_{i};\ \boldsymbol{\theta})$ in \eqref{optimization_EM_LS_represent}. 
Since the second term, i.e., $\overline{\mathbf{C}}^\mathsf{T}
\overline{\mathbf{c}}^{(i)}(\boldsymbol{\theta})^{*}/\lVert \overline{\mathbf{c}}^{(i)}(\boldsymbol{\theta})\rVert$ does not depend on $\overline{\mathbf{E}}_{i}$, it is calculated once at the beginning of the algorithm with complexity $O(N_gMS)$. 
For the first term, we use the fact that 
$
\overline{\mathbf{E}}_i^\mathsf{T}
\overline{\mathbf{E}}_i
=
\mathrm{diag}(\widetilde{\mathbf{e}}_i)
$, where 
$
\widetilde{\mathbf{e}}_i
=
[\overline{\mathbf{e}}_{1,i}^\mathsf{T}, \dots, \overline{\mathbf{e}}_{M,i}^\mathsf{T}]^\mathsf{T}
\in 
\{0,1\}^{MS}.
$
Thus, 
$
\overline{\mathbf{E}}_i^\mathsf{T}
\overline{\mathbf{E}}_i
\overline{\mathbf{c}}^{(i)}(\boldsymbol{\theta})^{*}
=
\widetilde{\mathbf{e}}_i
\odot 
\overline{\mathbf{c}}^{(i)}(\boldsymbol{\theta})^{*}.
$
Hence, updating the $m$-th antenna state i.e., modifying $\overline{\mathbf{e}}_{m,i}$, requires the complexity of $O(N_{g})$ to update the first term of $\mathcal{G}_{i}(\overline{\mathbf{E}}_{i};\boldsymbol{\theta})$.  
Consequently, line 6 of Algorithm~\ref{alg:EM_precoder_opt}, which evaluates all $S$ candidate states for antenna $m$, incurs $O(SN_{g})$ complexity. 
Sweeping over $M$ antennas in line 5 therefore costs $O(MSN_{g})$ per iteration. 
Over a maximum of $N_{\max}$ iterations, optimizing one codeword has complexity $O(N_{\max}MSN_{g})$, and executing the algorithm for all $N_{t}$ codewords yields a total complexity of $O(N_{t}N_{\max}MSN_{g})$.

\begin{remark}
The codebooks in \eqref{def:synthesis_codebook} and \eqref{def:finite_state_codebook} can be constructed entirely offline.  
Concretely, for each 2D-\ac{aod} inside the uncertainty region, we precompute the corresponding \ac{bb} and \ac{em} codewords on a discretized angular grid with $\theta^{\mathrm{az}}\in[-\pi,\pi)$ and $\theta^{\mathrm{el}}\in[0,\pi]$ at a chosen resolution. 
In online scenario, the \ac{bs} then selects the precomputed entries that cover the given uncertainty region, avoiding costly real-time optimization and enabling low-latency deployment.
\end{remark}
\vspace{-0.5cm}
\section{Maximum Likelihood Localization}\label{sec:ML_localization}
Due to the analogy between the received signals between the two considered reconfigurability paradigms in \eqref{eq:y-def-simp} and \eqref{eq:y-def-simp-discrete}, we only explain it for the synthesis scenario. 
The same methodology can be applied to finite-state selection model. 
The proposed approach consists of two steps: coarse positioning and refinement step, as detailed next.
\subsection{Step 1: Coarse Positioning}\label{sec:loc_step_1}
%\vspace{-1cm}
\subsubsection{Coarse Delay Estimation}
After substituting \eqref{def:w} in \eqref{eq:y-def-simp}, 
the received signal vector at the $t$-th transmission can be represented as:
\begin{equation}\label{eq:yt_model}
\mathbf{y}_t
=
\beta_t\,\mathbf{d}(\tau)
+
\widetilde{\mathbf{v}}_t,
\qquad 
t=1,\dots,N_t,
\end{equation}
where $\tau$ and 
$\beta_t=\sqrt{P}\,\alpha\,
\mathbf{c}(\boldsymbol{\theta})^{\mathsf{T}}\,
\mathbf{w}_t\in\mathbb{C}$ 
are the delay and unknown complex gain of the \ac{los} path, and $\widetilde{\mathbf{v}}_t\in \mathbb{C}^{N_s}$ denotes the joint contribution of the noise and interference.
The estimation problem can be formulated as follows:
\begin{equation}\label{eq:beta_tau_est_problem}
(\hat{\tau},\{\hat{\beta}_t\}_{t=1}^{N_t}) 
=
\underset{
\tau,\{\beta_t\}_{t=1}^{N_t}
}
{\textnormal{argmin}}
\sum_{t=1}^{N_t}
\lVert
\mathbf{y}_t - \beta_t\mathbf{d}(\tau)
\rVert^2.
\end{equation}
For a given fixed delay $\tau$, the \ac{ml} i.e., the \ac{ls} estimate of $\beta_t$ is obtained in closed form as:
\begin{equation}\label{eq:beta_hat_t}
\hat{\beta}_t(\tau) 
=
\frac{\mathbf{d}(\tau)^{\mathsf{H}}\mathbf{y}_t}{\lVert\mathbf{d}(\tau)\Vert^2}
=\frac{1}{N_s}\,\mathbf{d}(\tau)^{\mathsf{H}}\mathbf{y}_t,
\end{equation}
where we used the fact that $\mathbf{d}(\tau)^{\mathsf{H}}\mathbf{d}(\tau)=N_s$. 
Substituting \eqref{eq:beta_hat_t} into \eqref{eq:beta_tau_est_problem} and simplifying yields the following simplified delay estimation problem:
\begin{equation}\label{eq:tau_hat_ml}
\hat{\tau}
=
\underset{
\tau\in[\tau_{\min},\tau_{\max}]
}
{\textnormal{argmin}}
\sum_{t=1}^{N_t}
\lvert 
\mathbf{d}(\tau)^{\mathsf{H}}\mathbf{y}_t 
\rvert^2
=
\lVert 
\mathbf{d}(\tau)^{\mathsf{H}}\mathbf{Y}
\rVert^2,
\end{equation}
where $\mathbf{Y}=[\mathbf{y}_1,\dots ,\mathbf{y}_{N_t}]$, 
and $\tau_{\min},\tau_{\max}$ are the minimum and maximum possible delays of the \ac{los} path inside the uncertainty region $\bm{\mathcal{U}}$. 
Problem \eqref{eq:tau_hat_ml}, can be solved by performing a simple linear search over $N_\tau$ uniformly spaced grids in the in the interval $[\tau_{\min},\tau_{\max}]$.

\subsubsection{Coarse AOD Estimation}
After substituting the coarse delay estimate $\hat{\tau}$ and \eqref{eq:yt_model} back into \eqref{eq:beta_hat_t}:
\begin{equation}\label{eq:beta_hat_rep}
\hat{\beta}_t=
\frac{\mathbf{d}(\hat{\tau})^{\mathsf{H}}\mathbf{y}_t}{N_s}
=
\underbrace{
\frac{\sqrt{P}\,
\alpha\, 
\mathbf{d}(\hat{\tau})^\mathsf{H}\mathbf{d}(\tau)}{
N_s
}
}_{\kappa}
\mathbf{c}(\boldsymbol{\theta})^{\mathsf{T}}
\mathbf{w}_t
+
\underbrace{
\frac{\mathbf{d}(\hat{\tau})^{\mathsf{H}}\widetilde{\mathbf{v}}_t}{N_s}
}_{n_t}
\,
,
\end{equation}
where $\kappa$ is an unknown scalar, and $n_t$ is the resulted noise scalar after delay beamforming. 
Thus, after collecting the scalars $\hat{\beta}_t$ into a vector $\hat{\boldsymbol{\beta}}=[\hat{\beta}_1,\dots ,\hat{\beta}_{N_t}]^{\mathsf{T}}$, the estimation problem for $\boldsymbol{\theta},\kappa$ can be formulated as follows:
\begin{equation}\label{eq:theta_kappa_est_problem}
(\hat{\boldsymbol{\theta}},\hat{\kappa})
=
\underset{
\boldsymbol{\theta},\ \kappa
}
{\textnormal{argmin}}
\lVert 
\hat{\boldsymbol{\beta}}
-
\kappa\,
\mathbf{s}(\boldsymbol{\theta})
\rVert^2\,
,
\end{equation}
where
$\mathbf{s}(\boldsymbol{\theta})=\mathbf{c}(\boldsymbol{\theta})^{\mathsf{T}}
\mathbf{W}$ 
and 
$\mathbf{W}=[\mathbf{w}_1,\dots ,\mathbf{w}_{N_t}]$. 
For fixed $\boldsymbol{\theta}$, the closed form \ac{ls} estimate of $\kappa$ is obtained as:
$
\hat{\kappa}(\boldsymbol{\theta}) 
=
\frac{\mathbf{s}(\boldsymbol{\theta})^{\mathsf{H}}
\hat{\boldsymbol{\beta}}}{\lVert
\mathbf{s}(\boldsymbol{\theta})
\Vert^2}\,
.
$
Substituting $\hat{\kappa}(\boldsymbol{\theta}) $ into \eqref{eq:theta_kappa_est_problem} and simplifying yields the following simplified 2D-\ac{aod} estimation problem:
\begin{equation}\label{eq:coarse_aoa_est}
\hat{\boldsymbol{\theta}}
=
\underset{
\boldsymbol{\theta}
}
{\textnormal{argmin}}
\frac{
\lvert
\mathbf{s}(\boldsymbol{\theta})^{\mathsf{H}}
\hat{\boldsymbol{\beta}}
\rvert^2}{\lVert
\mathbf{s}(\boldsymbol{\theta})
\Vert^2}\,
,
\end{equation}
which can be solved via a simple 2D search over $N_{\theta}$ \ac{aod} grids inside uncertainty region $\bm{\mathcal{U}}$. 

Hence, combining the coarse delay and 2D-\ac{aod} estimates $\hat{\tau}$ and $\hat{\boldsymbol{\theta}}$, the coarse estimate of the \ac{ue} position is obtained as:
$
\tilde{\bm{p}}_{u}
=
c\,
\hat{\tau}\,
\hat{\mathbf{u}},\  
\hat{\mathbf{p}}_{u}
=
\bm{R}\,
\tilde{\bm{p}}_{u}
+
\bm{p}_b,
$
where $\hat{\bm{u}}$ is the unit direction towards the estimated 2D-\ac{aod} $\hat{\boldsymbol{\theta}}$, and $\tilde{\bm{p}}_{u}$ denotes the estimated position in the local coordinates of the \ac{er-fas}, and finally $\hat{\bm{p}}_{u}$ is the estimated \ac{ue} position in the global coordinates.
\vspace{-0.5cm}
\subsection{Step 2: Position Refinement}\label{sec:loc_step_2}
For the refinement step, we consider direct localization. 
In particular, the estimation problem can be formulated as:
\begin{equation}\label{opt:ML_problem}
(\hat{\zeta},\ \hat{\mathbf{p}}_u)
=
\underset{
\zeta,\ 
\mathbf{p}
}
{\textnormal{argmin}}
\quad  
\sum_{t=1}^{N_t}
\lVert
\mathbf{y}_t
-
\zeta
\mathbf{x}_t
\rVert^2,
\end{equation}
where $\zeta=\sqrt{P}\,\alpha$ is the unknown overall complex gain and
$\mathbf{x}_t=
\mathbf{d}(\tau)\,
\mathbf{c}(\boldsymbol{\theta})^\mathsf{T}
\mathbf{w}_t$. 
It is easy to observe that the \ac{ml} estimate of $\zeta$ minimizing \eqref{opt:ML_problem} is obtained as:
$
\hat{\beta}(\mathbf{p}_u)
=
\frac{
\widetilde{\mathbf{x}}^\mathsf{H}
\widetilde{\mathbf{y}}
}{
\lVert 
\widetilde{\mathbf{x}}
\rVert^2
}\ ,
$
where 
$\widetilde{\mathbf{x}}=[\mathbf{x}_1^\mathsf{T},\dots ,\mathbf{x}_{N_t}^\mathsf{T}]^{\mathsf{T}}$, and 
$\widetilde{\mathbf{y}}=[\mathbf{y}_1^\mathsf{T},\dots ,\mathbf{y}_{N_t}^\mathsf{T}]^{\mathsf{T}}$. 
After substituting $\hat{\beta}(\mathbf{p}_u)$ back into \eqref{opt:ML_problem}, the direct localization problem can be simplified as:
\begin{equation}\label{opt:ML_problem_simp}
\hat{\mathbf{p}}_u
=
\underset{
\mathbf{p}_u
}
{\textnormal{argmin}}
\quad  
\frac{
\lvert \widetilde{\mathbf{x}}^\mathsf{H}
\widetilde{\mathbf{y}}\rvert^2
}{
\lVert 
\widetilde{\mathbf{x}}
\rVert^2
}\ .
\end{equation}
\begin{comment}
We first discretize the uncertainty region $\mathcal{U}$ with step sizes $\Delta x$, $\Delta y$, and $\Delta z$ along the $x$, $y$, and $z$ axes, respectively, resulting in $N_{c}$ candidate points. 
The coarse position estimate $\hat{\mathbf{p}}_{u}$ is then chosen as the grid point that minimizes the cost in \eqref{opt:ML_problem_simp}.  
\end{comment}

Utilizing the initial position estimate in the previous subsection, we apply the derivative‐free \ac{nm} algorithm \cite{FADAKAR2024103382, fadakar2024deepdoa} to iteratively refine the \ac{ue} position. 
In particular, starting with the initial estimate $\hat{\mathbf{p}}_{u}$ in the previous subsection, this fast simplex-based algorithm searches $\mathcal{U}$ in an off-grid manner, leading to an accurate position estimate.
\vspace{-0.2cm}
\subsection{Complexity of the Proposed Localization Method}
The delay estimation problem \eqref{eq:tau_hat_ml} has the complexity $O(N_sN_tN_\tau)$, the complexity of \ac{aoa} estimation in \eqref{eq:coarse_aoa_est} is $O(N_tN_{\theta})$.
The computation of the objective function \eqref{opt:ML_problem_simp} has the complexity $O(N_tN_sMQ)$. 
Thus, the complexity of the second step, i.e., the position refinement, is given by $O(N_{\text{nm}}N_tN_sMQ)$, where $N_{\text{nm}}$ is the average number of iterations in the \ac{nm} algorithm until convergence. 
%\vspace{-0.5cm}
%\section{Complexity Analysis}
%In this section, we perform a comprehensive complexity analysis and obtain the complexity of each proposed methodology in terms of big-O notation.

\section{Simulations}\label{sec:simulations}
In this section, we evaluate the performance of the proposed codebook design and localization for \ac{er-fas}-assisted mmWave systems through numerical simulations.
\vspace{-0.5cm}
\subsection{Simulation Setup}\label{sec:simulation_setup}
The scenario under consideration includes a \ac{bs} equipped with a \ac{upa} with the same number of rows and columns $M^h=M^v$, and a single \ac{ue}. 
The default system parameters are presented in Table~\ref{tab:sys-params}. 
Note that some parameters may vary in different simulations. 
According to the selected uncertainty region $\mathbf{\mathcal{U}}$ in Table~\ref{tab:sys-params}, the elevation and azimuth intervals are uniformly sampled with a grid step of $d_\theta = d_\phi = \frac{1.8}{M^h}\,\mathrm{rad}$ from which the values of $L$ and $N_t$ are determined\footnote{These values are set according to the half-power beamwidth of the \ac{bs}.}. 
For instance, for $M^h=M^v=5$, $d_\theta = d_\phi\approx 20.63^\circ$ resulting in $L=3$ 2D-\ac{aod}s and $N_t=3L=9$ codewords. 
In the finite-state selection model, we utilize $N_g=1000$ angular grid points to optimize each \ac{em} precoder in solving \eqref{optimization_EM_LS}.
The channel amplitudes are modeled according to the well-known free-space path loss model 
$\rho=\frac{\lambda}{4\pi \lVert \mathbf{p}_b-\mathbf{p}_u\rVert}$ 
for the \ac{los} path and 
$\rho_i=\frac{\sqrt{4\pi s_i}\lambda}{16\pi^2 \lVert \mathbf{p}_b-\mathbf{p}_{s}^{(i)}\rVert\lVert \mathbf{p}_{s}^{(i)}-\mathbf{p}_u\rVert}$ 
for the $i$-th \ac{mpc}. 
The phase components are generated uniformly from the interval $[-\pi, \pi]$. 
The transmit power $P$ is varied to obtain different received \ac{snr} levels on the \ac{los} path, defined as $\mathrm{SNR}=\frac{P\rho^{2}}{N_{0}B}$ \cite{fascista2022ris}, where $N_{0}$ is the noise power spectral density, and $B$ is the bandwidth. 
We use \ac{rmse} and \ac{peb} metrics to evaluate the performance of the proposed methods. 
Each \ac{rmse} value is estimated using $1000$ Monte-Carlo trials. 
In all simulations, optimal power allocation is used unless stated otherwise. 
Moreover, except for the simulations reported in Sec.~\ref{sec:simul_interference}, we assume an interference-free environment.

%%\vspace{-0.5cm}
\begin{table}[ht]
%\captionsetup{font=scriptsize}
\caption{\label{tab:sys-params} System parameters}
\centering
\fontsize{12}{10}\selectfont % Set the font size to 12pt with a line spacing of 14pt
\resizebox{\columnwidth}{!}{
\begin{tabular}{ |l|l|  }
\hline
Default System Parameters and
Symbol
&
\textbf{Value}
\\
\hline
Carrier frequency $f_c$ & $30\,\mathrm{GHz}$
%\\
%Bandwidth & $W$ & $400\,\mathrm{MHz}$
\\
Noise PSD $N_0$ & $-173.855\,\mathrm{dBm}$
%\\
%Noise variance & $\sigma^2_n$ & $-87.8344\,\mathrm{dBm}$
\\
Light speed $c$ & $3\times 10^8\,\mathrm{m/s}$
\\
Subcarrier spacing $\Delta f$ & $200\,\mathrm{kHz}$
\\
Bandwidth $B$ & $100\,\mathrm{MHz}$
\\
UE position $\mathbf{p}_u$ & $[45,5,2]^\mathsf{T}$
\\
Number of \ac{shod} bases 
$Q$ &  $4$
\\
Uncertainty region $\mathbf{\mathcal{U}}$ & $30<x<50,-10<y<10$, $0<z<10$
\\
Elevation Angle Bounds
$[\theta^{\text{el}}_\text{min}, \theta^{\text{el}}_\text{max}]$ 
& 
$[80.53^\circ, 99.46^\circ]$
\\
Azimuth Angle Bounds 
$[\theta^{\text{az}}_\text{min}, \theta^{\text{az}}_\text{max}]$
&
$[-18.43^\circ, 18.43^\circ]$
\\
\ac{bs} position $\mathbf{p}_b$ & $[0,0,5]^\mathsf{T}\,\mathrm{[m]}$
\\
\ac{bs} number of elements $M^{h}$, $M^v$ & $5$, $5$
\\
Localization parameters $N_\tau$, $N_\theta$ & $1000$, $500$
\\
\hline
\end{tabular}
}
\end{table}
\vspace{-0.5cm}
\subsection{Beampattern Analysis}
In our first simulation, we compare the beampatterns of a representative codeword from our proposed synthesis codebook \eqref{def:synthesis_codebook} and finite‐state selection codebook \eqref{def:finite_state_codebook}, both evaluated at $\boldsymbol{\theta}=[90^\circ,0^\circ]^\mathsf{T}$.
The beampatterns are shown for all three types of codewords in Fig.~\ref{fig:patterns}.  
Here, we denote the beampattern for the synthesis model by $p(\boldsymbol{\theta}^{*})$ and for the finite‐state model by $\overline{p}(\boldsymbol{\theta}^{*})$, which are defined as:
\begin{equation}\label{def:beampattern}
p(\boldsymbol{\theta}^{*})
=
\frac{
\lvert\mathbf{c}(\boldsymbol{\theta}^{*})^\mathsf{T}\mathbf{w}\rvert^2}{\lVert \mathbf{w}\rVert^2}
,\ 
\overline{p}(\boldsymbol{\theta}^{*})
=
\frac{
\lvert\overline{\mathbf{c}}(\boldsymbol{\theta}^{*})^\mathsf{T}\overline{\mathbf{w}}\rvert^2}{\lVert \overline{\mathbf{w}}\rVert^2}.
\end{equation}
where $\boldsymbol{\theta}^{*}$ is an arbitrary 2D-\ac{aod} grid. 
%In Appendix~\ref{app:max_beampattern}, we show that if the number of \ac{shod} bases in the synthesis model is a perfect square, the maximum beampattern attains $MQ/(4\pi)$.  
%Consequently, the peak gain of an \ac{er-fas} with $Q$ bases exceeds that of a conventional non‐reconfigurable array by $10\log_{10}(Q)\,$dB. 
From Figs.~\ref{fig:pattern_1d_vs_theta_type_1} and \ref{fig:pattern_1d_vs_phi_type_1}, the \ac{er-fas} achieves approximately $13.98\,\mathrm{dB}$ higher peak beampattern gain than the traditional non-reconfigurable array for $S=64$ in finite-state selection model. 
Similarly, the peak of the beampattern using \ac{er-fas} with $Q=4$ in the synthesis model is $6\,\mathrm{dB}$ above the traditional array.

\begin{figure}
\centering
\begin{subfigure}{\columnwidth}
\centering
\includegraphics[width=\columnwidth]{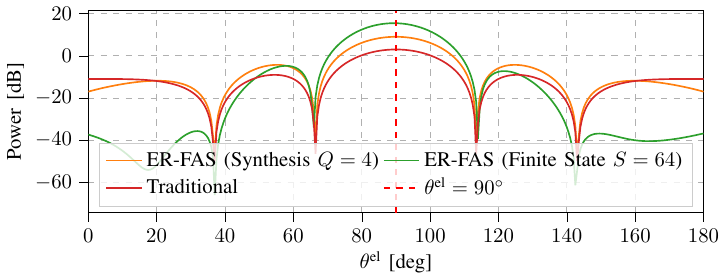}
\caption{1D beampattern of Type-1 codeword for fixed $\theta^{\text{az}}=0^\circ$}
\label{fig:pattern_1d_vs_theta_type_1}
\end{subfigure}%
\hfill
\vspace{6pt}
\begin{subfigure}{\columnwidth}
\centering
\includegraphics[width=\columnwidth]{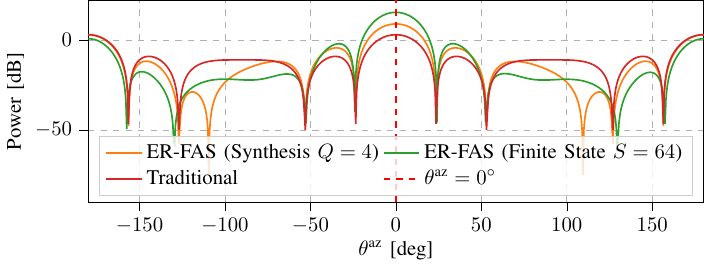}
\caption{1D beampattern of Type-1 codeword for fixed $\theta^{\text{el}}=90^\circ$}
\label{fig:pattern_1d_vs_phi_type_1}
\end{subfigure}%
\hfill
\vspace{6pt}
\begin{subfigure}{\columnwidth}
\centering
\includegraphics[width=\columnwidth]{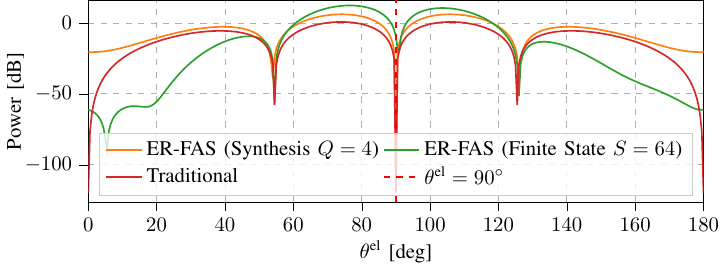}
\caption{1D beampattern of Type-2 codeword for fixed $\theta^{\text{az}}=0^\circ$}
\label{fig:pattern_1d_vs_theta_type_2}
\end{subfigure}%
\hfill
\vspace{6pt}
\begin{subfigure}{\columnwidth}
\centering
\includegraphics[width=\columnwidth]{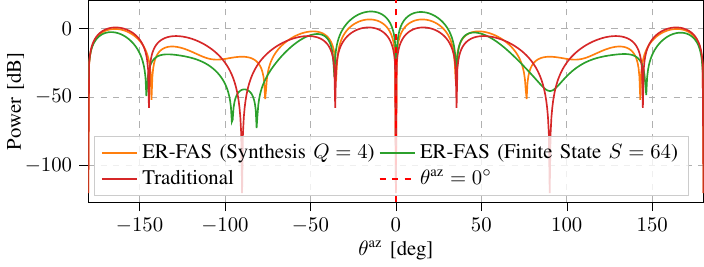}
\caption{1D Beampattern of Type-3 codeword at $\theta^{\text{el}}=90^\circ$, $Q=64$}
\label{fig:pattern_1d_vs_phi_type_3}
\end{subfigure}%
\caption{
1D beampatterns of the optimized codeword.
}
\label{fig:patterns}
\end{figure}
\vspace{-0.2cm}
\subsection{Performance Under Different Power Allocation Schemes}
In this subsection, we assess the performance of the optimal power allocation of the codewords proposed in Sec.~\ref{sec:power_alloc_synthesis} and Sec.~\ref{sec:power_alloc_d}. 
The uniform power allocation scheme is used as a baseline for comparative analysis. 
It is noteworthy that the uniform power allocation is a common heuristic approach in the state-of-the-art approaches \cite{fadakar2024multi, fadakar2025mutual} using traditional arrays. 
To this end, we obtain the \ac{peb} values for different \ac{ue} positions inside a square within $x\in[30,50]$ and $y\in[-10,10]$ at a fixed height $z=2$ and fixed $\mathrm{SNR}=5\,\mathrm{dB}$. 
As shown in Fig.~\ref{fig:crb_maps}, the proposed optimized power allocation consistently outperforms the uniform power allocation strategy, demonstrating its effectiveness.

\begin{figure}[ht]
\centering
\begin{subfigure}[t]{0.5\columnwidth}
\centering
\includegraphics[height=0.12\textheight]{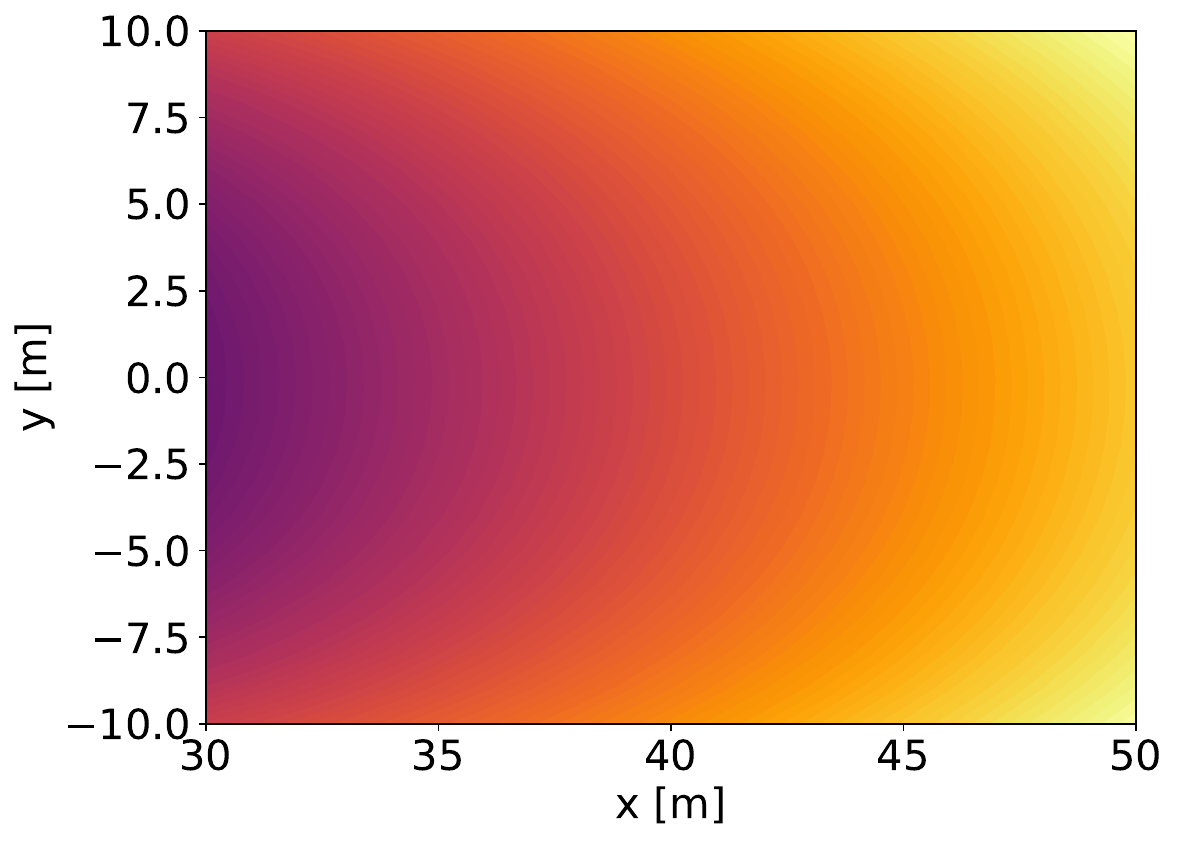}
\caption{Uniform power allocation}
\label{fig:crb_map_uniform}
\end{subfigure}%
\hfill
\begin{subfigure}[t]{0.5\columnwidth}
\centering
\includegraphics[height=0.12\textheight]{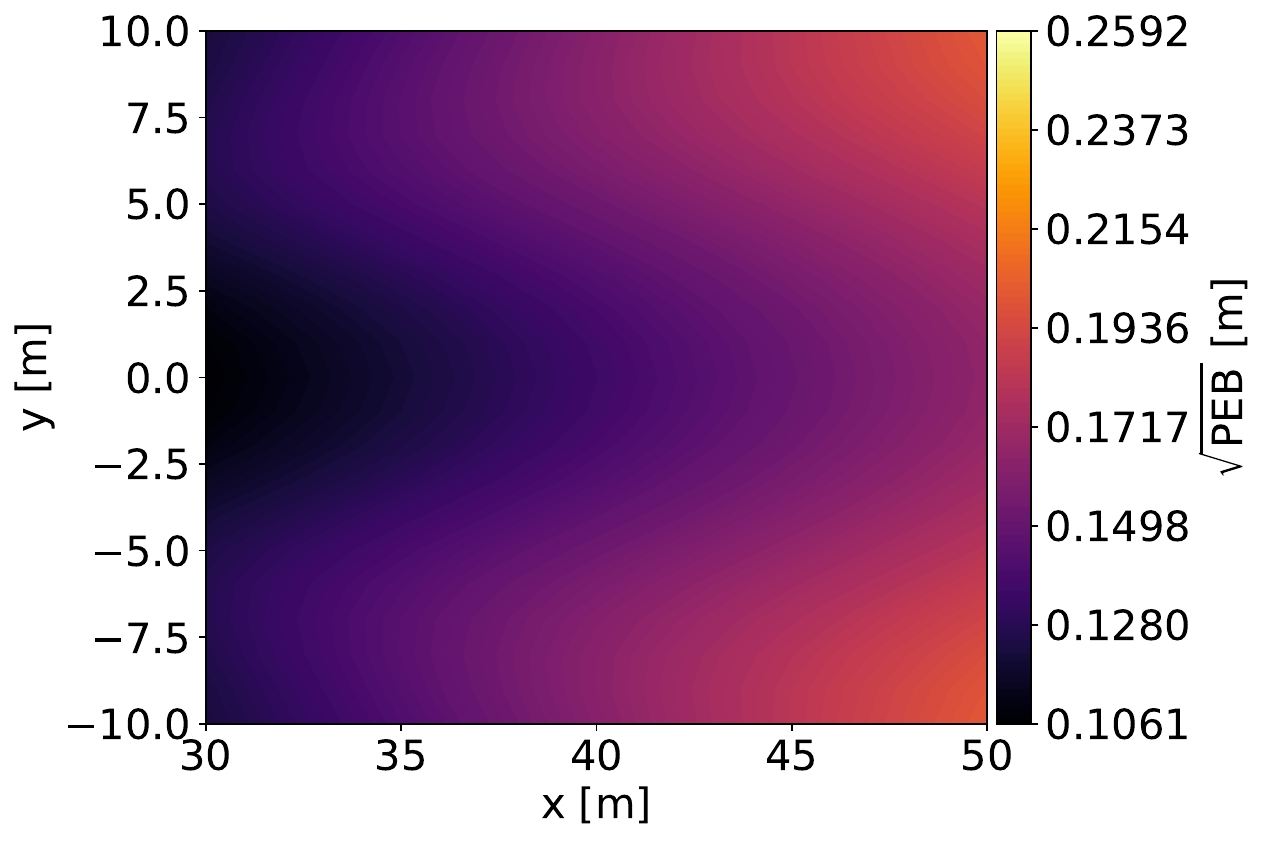}
\caption{Optimal power allocation}
\label{fig:crb_map_optimal}
\end{subfigure}
\caption{%
Evaluation of the proposed power allocation.
}
\label{fig:crb_maps}
\end{figure}

\vspace{-0.5cm}
\subsection{Localization Performance Versus SNR}
In Fig.~\ref{fig:RMSE_CRB_vs_Pow}, the \ac{rmse} and corresponding \ac{peb} of the proposed \ac{ml}-based localization from Sec.~\ref{sec:ML_localization} are plotted for three array configurations: 
the \ac{er-fas} under the synthesis model with $Q=4$ \ac{shod} bases, 
the \ac{er-fas} under the finite‐state selection model with $S=64$, 
and a traditional non‐reconfigurable array employing the optimal codebooks proposed in \cite{Keskin2022Optimal,fascista2022ris}. 
It is shown that both \ac{er-fas} configurations substantially outperform the traditional array in localization accuracy. 
Moreover, the \ac{er-fas} is more robust than a traditional array in low-\ac{snr} regimes.

\begin{figure}[!t]
\centering
\includegraphics[width=\columnwidth]{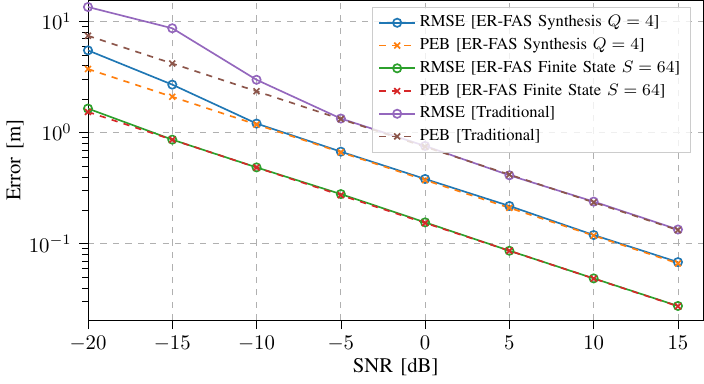}
\caption{
Localization performance versus \ac{snr}.
}
\label{fig:RMSE_CRB_vs_Pow}
\end{figure}
\vspace{-0.5cm}
\subsection{Impact of the Number of Bases and States}
Figs.~\ref{fig:CRB_vs_Q} and \ref{fig:CRB_vs_S} plot the \ac{peb} of the proposed method under the synthesis and finite‐state selection models versus the number of \ac{shod} bases $Q$ and states $S$, respectively, for various antenna geometries. 
The optimized traditional non-reconfigurable array is used as a baseline.
As expected, the performance of the proposed method improves with increasing $Q$ or $S$. 
However, the traditional array is fixed as it does not depend on either $Q$ or $S$.
Additionally, Fig.~\ref{fig:CRB_vs_Q} shows that the optimal traditional array coincides with the \ac{er-fas} synthesis design when $Q=1$. 
This equivalence follows because the first \ac{shod} basis is constant on the unit sphere. 
Consequently, the synthesis construction with $Q=1$ reduces to the optimal traditional array. 
See Appendix~\ref{app:shod_omni} for proof.
Moreover, Fig.~\ref{fig:CRB_vs_S} shows that a $5\times5$ \ac{upa} \ac{er-fas} outperforms a traditional $10\times10$ \ac{upa} for $S\ge 19$, achieving equivalent localization accuracy with significantly fewer elements, thereby simplifying practical deployment. 
Analogous observations hold for other array geometries.

\begin{figure}[!t]
\centering
\includegraphics[width=\columnwidth]{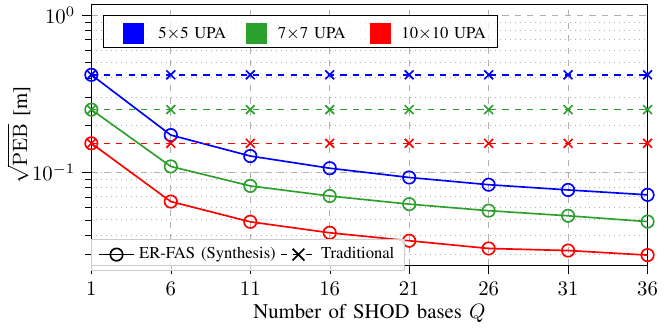}
\caption{
Localization performance versus number of \ac{shod} bases.
}
\label{fig:CRB_vs_Q}
\end{figure}

\begin{figure}[!t]
\centering
\includegraphics[width=\columnwidth]{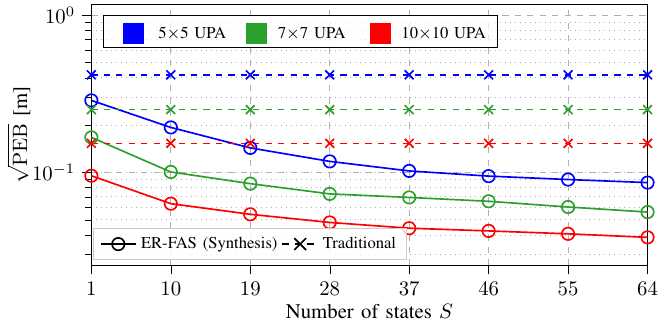}
\caption{
Localization performance versus number of states.
}
\label{fig:CRB_vs_S}
\end{figure}

\vspace{-0.5cm}
\subsection{Localization Performance Under Interference}\label{sec:simul_interference}
In this subsection, we evaluate the robustness of the proposed methods in the presence of interference. 
We also consider the traditional non-reconfigurable array as a baseline. 
To this end, we consider $I=40$ \ac{mpc_p} to model interference. 
The positions of these \ac{mpc_p} are randomly generated inside the \ac{ue} uncertainty region $\bm{\mathcal{U}}$. 
Then, we obtain the \ac{rmse} and \ac{peb} metrics versus \ac{lmr} \cite[Eq.~(24)]{fadakar2024multi} at a fixed $\mathrm{SNR}=0\,\mathrm{dB}$ in the interval $\mathrm{LMR}\in[0\,\mathrm{dB},45\,\mathrm{dB}]$ as shown in Fig.~\ref{fig:RMSE_CRB_vs_LMR}. 
Each \ac{rmse} value is computed using $1000$ Monte-Carlo trials.
As discussed in Sec.~\ref{sec:fim_ch_domain}, since the \ac{fim} analysis in this paper considers only \ac{los} path, the \ac{peb} is fixed for each \ac{lmr}. 
However, \ac{rmse} converge to the \ac{peb} at approximately $\mathrm{LMR}=15\,\mathrm{dB}$, and the proposed \ac{er-fas} is close to the \ac{peb} in $\mathrm{LMR}\in[10\,\mathrm{dB},15\,\mathrm{dB}]$. 
It is evident that the proposed method using \ac{er-fas} substantially outperforms the traditional non-reconfigurable array even in challenging low-\ac{lmr} regimes.
The performance degradation in low \ac{lmr} is due to the strong multipath components and weak \ac{los} path, resulting a model mismatch thereby degrading the localization performance. 
However, it is noteworthy that based on mmWave channel measurements, the power of the \ac{los} component is shown to be around $13\,\mathrm{dB}$ greater than that of the \ac{nlos} components \cite{muhi2010modelling, Wang2021Joint_beam}. 
Thus, the proposed method remains robust in practical scenarios.

\begin{figure}[!t]
\centering
\includegraphics[width=\columnwidth]{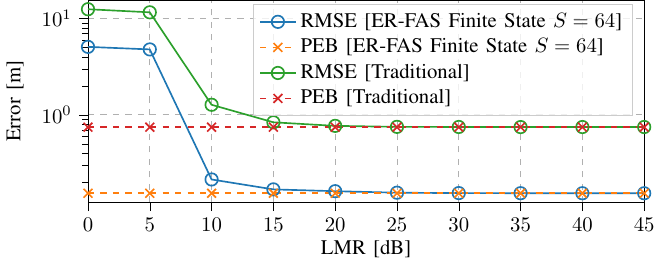}
\caption{
Localization performance versus \ac{lmr}.
}
\label{fig:RMSE_CRB_vs_LMR}
\end{figure}

\section{Conclusion}
In this paper, we investigated the joint design of \ac{bb} and \ac{em} precoders to maximize downlink localization accuracy in an \ac{er-fas}-enabled \ac{miso} system. 
We considered two reconfigurability paradigms: (i) a synthesis model, where each antenna synthesizes beampatterns from a set of orthonormal basis functions, and (ii) a finite-state selection model, where each antenna chooses from a library of predefined patterns. 
For both paradigms, we proposed low-complexity codebooks for the \ac{bb} and \ac{em} designs.
For the finite-state model, we developed an efficient \ac{bcd} algorithm to optimize the \ac{em} precoders. 
We complemented our designs with an efficient \ac{ml}-based localization algorithm and extensive simulations, demonstrating substantial \ac{peb} and \ac{rmse} gains over traditional non-reconfigurable arrays. 
Future work will extend these methods to \ac{isac} systems and assess robustness to model mismatch and hardware impairments.  
\vspace{-0.5cm}
\appendices
\section{Complex Spherical Harmonic Orthogonal Decomposition}
\label{app:shod}
\ac{shod} bases represent any square‑integrable function on the unit sphere as a weighted sum of orthonormal basis functions \cite{Costa2010Unified}.  
In antenna and array processing, \ac{shod} provides a compact way to model 3D radiation patterns.  
Each basis function is indexed by a nonnegative integer degree $\ell$ and an integer order $m$ with $-\ell\le m\le\ell$.  The degree $\ell$ controls the overall spatial frequency (number of lobes), while the order $m$ governs azimuthal variation.  Low degrees (e.g.\ $\ell=0$ or $\ell=1$) yield broad, smooth patterns; higher $\ell$ produce finer angular structure and more nulls.
%
%As mentioned in early sections, in practice we truncate to the first $Q$ bases (for some $Q\in\mathbb{N}$), selected in e.g.\ increasing $(\ell,m)$ order.  This yields a finite dictionary of $Q$ functions.
\vspace{-0.5cm}
\subsection{Definition of the Complex SHOD Bases}
The complex spherical harmonics are defined for each degree $\ell\ge0$ and order $-\ell\le m\le\ell$ by
\begin{equation}
Y_\ell^m(\theta^{\text{el}}, \theta^{\text{az}})
= (-1)^m\,N_{\ell m}\,P_\ell^{m}(\cos\theta^{\text{el}})\,e^{\,j m\theta^{\text{az}}},
\label{eq:Ylm_complex}
\end{equation}
where the normalization constant is
$
N_{\ell m}
=\sqrt{\frac{2\ell+1}{4\pi}\,\frac{(\ell-m)!}{(\ell+m)!}},
$
and $P_\ell^{m}(x)$ denotes the associated Legendre function.
These functions satisfy the orthonormality relation
\begin{align}
\int_{-\pi}^{\pi}\!\!\int_{0}^{\pi}
& Y_\ell^m(\theta^{\text{el}}, \theta^{\text{az}})\,
\bigl(Y_{\ell'}^{m'}(\theta^{\text{el}}, \theta^{\text{az}})\bigr)^*
\,
\notag \\
&
\times 
\sin\theta^{\text{el}}\,d\theta^{\text{el}}\,d\theta^{\text{az}}
=\delta_{\ell\ell'}\,\delta_{mm'},
\end{align}
where $\delta_{pq}$ is the Kronecker delta, equal to $1$ if $p=q$ and $0$ otherwise.
To form a truncated dictionary of size $Q$, we select the first $Q$ pairs $(\ell,m)$ in the ordering
$
(\ell,m)=(0,0),\,(1,-1),\,(1,0),\,(1,1),\,(2,-2),\dots
$
and denote the $k$-th base by 
$\,Y_k(\theta^{\text{el}},\theta^{\text{az}})$.

\vspace{-0.5cm}

\subsection{Special Case: Recovering the Omni‑Directional Pattern}
\label{app:shod_omni}
The degree‑zero spherical harmonic is
$
Y_0^0(\boldsymbol{\theta})
= N_{0,0}\,P_0^0(\cos\theta^{\text{el}})
= \sqrt{\frac{1}{4\pi}},
$
a constant over the sphere.  Hence, by choosing the \ac{em} precoder vector
$
\mathbf{e} = \bigl[1,\,0,\dots,0\bigr]^\mathsf{T},
$
so that
$
g(\boldsymbol{\theta}) = \mathbf{e}^\mathsf{H}\mathbf{b}(\boldsymbol{\theta})
= Y_0^0(\boldsymbol{\theta}) = \frac{1}{\sqrt{4\pi}},
$
we can recover an \emph{omni‑directional} element radiation pattern.  

%In other words, a conventional omni‑antenna is obtained by selecting only the $\ell=0,m=0$ basis and zeroing all higher modes. 
\vspace{-0.5cm}
\subsection{Total Radiated Power Conservation}
\label{app:power_conservation}
Consider the full 2D angular domain $\boldsymbol{\theta}=[\theta^{\text{el}},\theta^{\text{az}}]^\mathsf{T}$ with $\theta^{\text{el}}\in[0,\pi]$, $\theta^{\text{az}}\in[-\pi,\pi)$. 

\subsubsection{Case 1: Synthesis Model}
Since the bases $\mathbf{b}(\boldsymbol{\theta})$ are orthonormal over the sphere,
\begin{equation}
\int_{0}^{2\pi}\!\int_{0}^{\pi}
\mathbf{b}(\boldsymbol{\theta})\,
\mathbf{b}(\boldsymbol{\theta})^H
\;\sin\theta^{\text{el}}\,
d\theta^{\text{el}}\,d\theta^{\text{az}}
\;=\;\mathbf{I}_Q,
\end{equation}
Then, the total radiated power from an arbitrary antenna with a unit norm \ac{em} vector $\mathbf{e}$ can be calculated as:
\begin{align*}
&\int_{0}^{2\pi}\!\int_{0}^{\pi}
\lvert
\mathbf{e}^\mathsf{H} \mathbf{b}(\boldsymbol{\theta})
\rvert^2
\sin\theta^{\text{el}}\,
d\theta^{\text{el}}\,d\theta^{\text{az}}\\
&\quad=
\int_{0}^{2\pi}\!\int_{0}^{\pi}
\mathbf{e}^\mathsf{H}
\big[
\mathbf{b}(\boldsymbol{\theta})b(\boldsymbol{\theta})^H\bigr]
\mathbf{e}
\sin\theta^{\text{el}}\,
d\theta^{\text{el}}\,d\theta^{\text{az}}\\
&\quad=
\mathbf{e}^\mathsf{H}
\Bigl[\int_{0}^{2\pi}\!\int_{0}^{\pi}
\mathbf{b}\,\mathbf{b}^\mathsf{H}\,\sin\theta^{\text{el}}
\,d\theta^{\text{el}}\,d\theta^{\text{az}}\Bigr]
\mathbf{e}
= 
\mathbf{e}^\mathsf{H} 
\mathbf{I}_Q 
\mathbf{e}
=
1.
\end{align*}

\subsubsection{Case 2: Finite-State Selection Model}
Since $\lVert\overline{\mathbf{e}}_{m,t}\rVert_{0}=1$, we have 
$\bigl[\mathbf{g}_{t}(\boldsymbol{\theta})\bigr]_{m}
=
\overline{\mathbf{e}}_{m,t}^{\mathsf{T}}
\,
\overline{\mathbf{b}}(\boldsymbol{\theta})=\overline{b}_s(\boldsymbol{\theta})$ for some $s\in\{1,\dots,S\}$.
Thus:
\begin{equation}
\begin{split}
&\int_{0}^{2\pi}\!\int_{0}^{\pi}
\big\lvert\bigl[\mathbf{g}_{t}(\boldsymbol{\theta})\bigr]_{m}\big\rvert^2
\sin\theta^{\text{el}}\,
d\theta^{\text{el}}\,d\theta^{\text{az}}\\
&\quad=
\int_{0}^{2\pi}\!\int_{0}^{\pi}
\lvert\overline b_{s}(\boldsymbol{\theta})\rvert^2
\sin\theta^{\text{el}}\,
d\theta^{\text{el}}\,d\theta^{\text{az}}
=1.
\end{split}
\end{equation}

In both models the total radiated power over the full sphere remains $1$, confirming the energy conservation law \cite{Ying2024Reconfigurable}.

\vspace{-0.3cm}
\section{Proof of Prop.~\ref{prop:opt_structure}}\label{app:opt_structure_proof}
Inspired by \cite{Li2008Range}, $\mathbf{W}\succeq 0$ can be represented as:
\begin{equation}\label{eq:W_Q}
\mathbf{W}
=
\mathbf{Q}
\mathbf{Q}^\mathsf{H}.
\end{equation}
Next, we decompose $\mathbf{Q}$ as
$
\mathbf{Q}
=
\mathbf{\Pi}_{\mathbf{C}_{w}}
\mathbf{Q}
+
\mathbf{\Pi}_{\mathbf{C}_{w}}^{\perp}
\mathbf{Q}
$, and substitute it to \eqref{eq:W_Q} to obtain the following decomposition for $\mathbf{W}$:
\begin{equation}\label{eq:W_decompose}
\mathbf{W}
=
\mathbf{\Pi}_{\mathbf{C}_{w}}
\mathbf{Q}
\mathbf{Q}^\mathsf{H}
\mathbf{\Pi}_{\mathbf{C}_{w}}
+
\widetilde{\mathbf{W}},
\end{equation}
where
\begin{equation}\label{eq:W_tilde_def}
\widetilde{\mathbf{W}}
=
\mathbf{\Pi}_{\mathbf{C}_{w}}^{\perp}
\mathbf{Q}
\mathbf{Q}^\mathsf{H}
\mathbf{\Pi}_{\mathbf{C}_{w}}^{\perp}
+
\mathbf{\Pi}_{\mathbf{C}_{w}}
\mathbf{Q}
\mathbf{Q}^\mathsf{H}
\mathbf{\Pi}_{\mathbf{C}_{w}}^{\perp}
+
\mathbf{\Pi}_{\mathbf{C}_{w}}^{\perp}
\mathbf{Q}
\mathbf{Q}^\mathsf{H}
\mathbf{\Pi}_{\mathbf{C}_{w}}.
\end{equation}

It can be easily verified that 
\begin{equation}\label{eq:W_tilde_dependency}
\mathbf{C}_{w}^\mathsf{H}
\widetilde{\mathbf{W}}
\mathbf{C}_{w}
=
0.
\end{equation}
On the other hand, according to Lemma~\ref{lemma:W_affine}, the \ac{fim} elements only depend on the elements of the matrix $\mathbf{C}_{w}^\mathsf{H}
\mathbf{W}
\mathbf{C}_{w}\in \mathbb{C}^{3\times 3}$. 
Hence, according to \eqref{eq:W_decompose} and \eqref{eq:W_tilde_dependency}, \ac{fim} does not depend on the component $\widetilde{\mathbf{W}}$ of $\mathbf{W}$. 
Moreover, from \eqref{eq:W_tilde_def}:
\begin{equation}\label{eq:W_tilde_power}
\rm tr(\widetilde{\mathbf{W}})
=
\rm tr(
\mathbf{\Pi}_{\mathbf{C}_{w}}^{\perp}
\mathbf{Q}
\mathbf{Q}^\mathsf{H}
\mathbf{\Pi}_{\mathbf{C}_{w}}^{\perp}
)
=
\lVert 
\mathbf{Q}^\mathsf{H}
\mathbf{\Pi}_{\mathbf{C}_{w}}^{\perp}
\rVert_{F}^2
\ge 0.
\end{equation}
Thus, we conclude that $\widetilde{\mathbf{W}}$ is a component of $\mathbf{W}$ for which the \ac{fim} (and hence \ac{peb}) is not dependent on, and moreover, it contains nonnegative portion of the total transmitted power. 
Thus, for optimal solution $\mathbf{W}$ we must have $\rm tr(\widetilde{\mathbf{W}})=0$ or equivalently $\mathbf{Q}^\mathsf{H}\mathbf{\Pi}_{\mathbf{C}_{w}}^{\perp}=0$ from \eqref{eq:W_tilde_power}. 
Because otherwise, one can find a better solution by considering $\widetilde{\mathbf{W}}=0$ (or equivalently reducing the power of $\widetilde{\mathbf{W}}$ to zero) and scaling the power of the first component of $\mathbf{W}$ in \eqref{eq:W_decompose}, which would yield a better solution with lower \ac{peb}. 
Hence, we conclude that an optimal solution $\mathbf{W}$ can be represented as follows:
\begin{align}
\mathbf{W}
& =
\mathbf{\Pi}_{\mathbf{C}_{w}}
\mathbf{Q}
\mathbf{Q}^\mathsf{H}
\mathbf{\Pi}_{\mathbf{C}_{w}}
\notag \\
& =
\mathbf{C}_{w}
\underbrace{
(\mathbf{C}_{w}^{\mathsf{H}}\mathbf{C}_{w})^{-1}
\mathbf{C}_{w}^{\mathsf{H}}
\mathbf{Q}
\mathbf{Q}^\mathsf{H}
\mathbf{C}_{w}
(\mathbf{C}_{w}^{\mathsf{H}}\mathbf{C}_{w})^{-1}
}_{\stackrel{\text{def}}{=}\boldsymbol{\Xi}}
\mathbf{C}_{w}^{\mathsf{H}}
\notag \\
& =
\mathbf{C}_{w}
\boldsymbol{\Xi}
\mathbf{C}_{w}^{\mathsf{H}},
\end{align}
where $\boldsymbol{\Xi}\in\mathbb{C}^{3\times 3}$ and $\boldsymbol{\Xi}\succeq 0$, which completes the proof.

\vspace{-0.5cm}
\section{Proof of Prop.~\ref{prop:opt_precoders}}\label{app:opt_precoders}
According to Prop.~\ref{prop:opt_structure}:
\begin{align}\label{eq:w_xi}
\mathbf{W}
& =
\sum_{t=1}^{N_t}
\mathbf{W}_t
=
\sum_{t=1}^{N_t}
\mathbf{w}_t
\mathbf{w}_t^\mathsf{H}
\notag \\
& =
\mathbf{C}_w
\mathbf{\Xi}
\mathbf{C}_w^\mathsf{H}
=
\sum_{i=1}^{3}\sum_{j=1}^{3}
[\mathbf{\Xi}]_{i,j}
[\mathbf{C}_w]_{:,i}
[\mathbf{C}_w]_{:,j}^{\mathsf{H}}.
\end{align}
Thus, we should first optimize the elements of the matrix $\mathbf{\Xi}$ to minimize the objective function \eqref{opt_prob_represent}, and then find suitable parameter $N_t$ and vectors $\{\mathbf{w}_t\}_{t=1}^{N_t}$ to satisfy \eqref{eq:w_xi}. 
Observing \eqref{eq:w_xi}, in order to make the problem easier and more tractable, we relax the problem by restricting $\mathbf{\Xi}$ to be a diagonal matrix $\mathbf{\Xi}=\rm diag(\boldsymbol{\xi})$ for some vector $\boldsymbol{\xi}\in\mathbb{R}_{+}^{3}$ with nonnegative elements. 
After substituting this equation in \eqref{eq:w_xi}:
\begin{equation}
\sum_{t=1}^{N_t}
\mathbf{w}_t
\mathbf{w}_t^\mathsf{H}
=
\sum_{i=1}^{3}
[\mathbf{\xi}]_{i}
[\mathbf{C}_w]_{:,i}
[\mathbf{C}_w]_{:,i}^{\mathsf{H}},
\end{equation}
which admits the following closed-form solutions:
\begin{equation}\label{eq:w_sol}
N_t=3,\ 
\mathbf{w}_i
=
\sqrt{[\mathbf{\xi}]_{i}}
[\mathbf{C}_w]_{:,i},
\end{equation}
for $i=1,2,3$. 
It is easy to see that $[\mathbf{\xi}]_{i}$ is proportional to how much power should be allocated for $\mathbf{w}_i$. 
Thus, for convenience we define 
$
[\mathbf{\xi}]_{i}
=
\delta_i/\lVert [\mathbf{C}_w]_{:,i}\rVert^2,
$
where $1\ge \delta_i\ge 0$ determines the exact portion of the total transmitted power $P$ in \eqref{eq:y-def} dedicated for the $i$-th codeword. Thus, the first part of the proposition is proved. 

Next, to obtain the corresponding \ac{bb} and \ac{em} precoders, we use the equation \eqref{def:w} which establishes a relation between the vector $\mathbf{w}_i$ and $\mathbf{E}_i^\mathsf{T}, 
\mathbf{f}_{i}$. 
According to the definition of $\mathbf{E}_i$, we obtain the following equation:
\begin{equation}\label{eq:w_f_e}
[\mathbf{w}_i]_{(m-1)Q+1:mQ}
=
[\mathbf{f}_{i}]_m
\mathbf{e}_{m,i},
\end{equation}
for $m=1,\dots ,M$. 
After substituting \eqref{eq:w_sol} in \eqref{eq:w_f_e}:
$
[\mathbf{f}_{i}]_m
\mathbf{e}_{m,i}
=
\frac{\delta_i \mathbf{c}_{m}^{(i)}(\boldsymbol{\theta})}{\lVert [\mathbf{c}^{(i)}(\boldsymbol{\theta})\rVert}.
$
Since $\lVert \mathbf{e}_{m,i}\rVert^2 = 1$, we obtain the following unique family of solutions for $[\mathbf{f}_{i}]_m$ and $\mathbf{e}_{m,i}$:
\begin{align}
[\mathbf{f}_{i}]_m
& =
\sqrt{\delta_i}
\frac{\lVert\mathbf{c}_{m}^{(i)}(\boldsymbol{\theta})\rVert}{\lVert \mathbf{c}^{(i)}(\boldsymbol{\theta})\rVert}
e^{j \psi_{m,i}}
,\ 
\mathbf{e}_{m,i}
=
\frac{\mathbf{c}_{m}^{(i)}(\boldsymbol{\theta})}{\lVert \mathbf{c}_{m}^{(i)}(\boldsymbol{\theta})\rVert}
e^{-j \psi_{m,i}},
\end{align}
where $\{\psi_{m,i}\}_{m=1,i=1}^{M,3}$ are arbitrary phases. 
\vspace{-0.5cm}
\section{Proof of Prop.~\ref{prop:optimal_precoder_d}}\label{app:prop3_proof}
By replacing $\mathbf{q}_t(\boldsymbol{\theta})$ and $\mathbf{f}_t$ in \eqref{eq:y-def-simp-d-v1} with $\mathbf{c}(\boldsymbol{\theta})$ and $\mathbf{w}_t$, the problem becomes similar to the synthesis case discussed in \ref{sec:synthesis_optimization}. 
Thus, by analogy, in a similar manner to Prop~\ref{prop:opt_structure} and Prop.~\ref{prop:opt_precoders} (first part) for the synthesis case, we choose the following three codewords to achieve an approximate optimal performance:
\begin{equation}\label{eq:bb_opt_fixed_em}
\hat{\mathbf{f}}_i
=
\sqrt{\delta_i}
\frac{\mathbf{q}^{(i)}(\boldsymbol{\theta})^{*}}{\lVert \mathbf{q}^{(i)}(\boldsymbol{\theta})\rVert}
,
\end{equation}
for $i=1,2,3$, 
where
\begin{equation}\label{eq:q_partials}
\mathbf{q}^{(1)}(\boldsymbol{\theta})
=
\mathbf{q}_1(\boldsymbol{\theta})
,\ 
\mathbf{q}^{(2)}(\boldsymbol{\theta})
=
\frac{\partial\,\mathbf{q}_2(\boldsymbol{\theta})}{\partial\theta^{\text{el}}}
,\ 
\mathbf{q}^{(3)}(\boldsymbol{\theta})
=
\frac{\partial\,\mathbf{q}_3(\boldsymbol{\theta})}{\partial\theta^{\text{az}}}
\end{equation}

After substituting \eqref{eq:q_represent_discrete} in \eqref{eq:q_partials} and using \eqref{def:c_partials_d},
the equations \eqref{eq:digital_opt_structure_d} are obtained, which completes the proof.

\vspace{-0.3cm}
\bibliographystyle{IEEEtran}
\bibliography{IEEEabrv,bib}{}
%\printbibliography
\end{document}